\documentclass{lmcs}


\usepackage{mathrsfs}

\usepackage{graphicx}
\usepackage{collect}
\usepackage{float}
\usepackage{tipa}
\usepackage{bbm}
\usepackage{stackengine}
\usepackage[title]{appendix}
\usepackage{etoolbox}
\usepackage{setspace}
\usepackage[most]{tcolorbox}
\usepackage{amsmath}
\usepackage{amsthm}
\usepackage{amssymb}
\usepackage{amsfonts}
\usepackage{pst-plot}
\usepackage{stackengine}
\usepackage{mathtools}
\usepackage{stmaryrd}
\usepackage{bussproofs}
\usepackage{caption}
\usepackage{chapterfolder}
\makeatletter
\let\th@plain\relax
\usepackage{tikz}
\usetikzlibrary{arrows}
\usepackage{tikz-cd}
\usepackage{cmll}
\usepackage{hyperref}
\usepackage{syntax}
\usepackage[shortlabels]{enumitem}

\renewcommand{\epsilon}[0]{\varepsilon}

\newcommand{\LHS}[0]{\text{LHS}}
\newcommand{\RHS}[0]{\text{RHS}}

\newcommand{\NN}[0]{\mathbb{N}}

\newcommand{\D}[0]{\text{D}}

\newcommand{\brak}[1]{\left[ #1 \right]}

\newcommand{\pa}[1]{\left( #1 \right)}

\newcommand{\ra}[0]{\rightarrow}

\newcommand{\da}[0]{\downarrow}

\newcommand{\Ra}[0]{\Rightarrow}

\newcommand{\defeq}[0]{\coloneqq}

\newif\ifdraft
\draftfalse

\renewcommand{\red}[1]{{\color{red} {#1}}}

\ifdefined\todo
\renewcommand{\todo}[1]{
  \red{[TODO: #1]}
}
\else
\newcommand{\todo}[1]{
  \red{[TODO: #1]}
}
\fi

\newcommand{\cat}[1]{\mathbf{#1}}

\newcommand{\homset}[3]{{#1}\left[{#2}, {#3}\right]}

\newcommand{\arr}[1]{\mathbf{#1}}
\newcommand{\A}[1]{\arr{#1}}
\newcommand{\Id}[1]{\arr{id}_{#1}}

\newcommand{\terminal}[0]{\mathbf{1}}
\newcommand{\pair}[2]{\left\langle {#1}, {#2} \right\rangle}
\newcommand{\four}[4]{\pair{\pair{#1}{#2}}{\pair{#3}{#4}}}

\makeatletter
\newsavebox{\@brx}
\newcommand{\llangle}[1][]{\savebox{\@brx}{\(\m@th{#1\langle}\)}%
  \mathopen{\copy\@brx\kern-0.5\wd\@brx\usebox{\@brx}}}
\newcommand{\rrangle}[1][]{\savebox{\@brx}{\(\m@th{#1\rangle}\)}%
  \mathclose{\copy\@brx\kern-0.5\wd\@brx\usebox{\@brx}}}
\makeatother

\newcommand{\ul}[1]{\underline{#1}}

\newcommand{\alter}[0]{\ |\ }
\newcommand{\reducesTo}[0]{\leadsto}
\newcommand{\parTo}[0]{
  \mathrel{\mathrlap{\raisebox{2pt}{$\reducesTo$}}
    \raisebox{-2pt}{$\reducesTo$}}}
\newcommand{\subst}[3]{{#1}\left[{#3}/{#2}\right]}
\newcommand{\diffS}[3]{\frac{\partial {#1}}{\partial {#2}}\pa{{#3}}}
\DeclarePairedDelimiter{\ceil}{\lceil}{\rceil}

\newcommand{\ts}[0]{\vdash}
\newenvironment{bprooftree}
  {\leavevmode\hbox\bgroup}
  {\DisplayProof\egroup}
  
\newcommand{\bpt}[1]{
  \begin{bprooftree}
    #1
  \end{bprooftree}
}

\newcommand{\sem}[1]{{\left\llbracket {#1} \right\rrbracket}}

\DeclareMathOperator*{\boxsum}{\boxplus}
\DeclareMathOperator*{\cansum}{\boxplus}
\newcommand{\continued}[1][30pt]{&\hspace{#1}}

\newcommand{\dd}[2][]{\partial#1\left[{#2}\right]}

\newcommand{\CdCAx}[1]{{{\bf [C$\partial$.{#1}]}}}

\newtheorem{conjecture}{Conjecture}[section]

\definecolor{indigo}{HTML}{332288} 
\definecolor{cyan}{HTML}{88ccee}
\definecolor{teal}{HTML}{44aa99}   
\definecolor{green}{HTML}{117733}  
\definecolor{olive}{HTML}{999933}
\definecolor{sand}{HTML}{ddcc77}   
\definecolor{rose}{HTML}{cc6677}
\definecolor{wine}{HTML}{882255}   
\definecolor{purple}{HTML}{aa4499} 
\definecolor{Orchid}{rgb}{0.6, 0.2, 0.8}

\lstdefinelanguage{Coq}{
morekeywords=[1]{Section, Module, End, Require, Import, Export,
  Variable, Variables, Parameter, Parameters, Axiom, Hypothesis,
  Hypotheses, Notation, Local, Tactic, Reserved, Scope, Open, Close,
  Bind, Delimit, Definition, Let, Ltac, Fixpoint, CoFixpoint,
  Morphism, Relation, Implicit, Arguments, Unset, Contextual,
  Strict, Prenex, Implicits, Inductive, CoInductive, Record,
  Structure, Canonical, Coercion, Context, Class, Global, Instance,
  Program, Infix, Theorem, Lemma, Corollary, Proposition, Fact,
  Remark, Example, Proof, Goal, Save, Qed, Defined, Hint, Resolve,
  Optimize,
  Rewrite, View, Search, Show, Print, Printing, All, Eval, Check,
  Projections, inside, outside, Def},
morekeywords=[2]{forall, exists, exists2, fun, fix, cofix, struct, match,
  lazymatch, context, with, end, as, in, return, let, if, is, then, else, for, of,
  nosimpl, when},
morekeywords=[3]{Type, Prop, Set, true, false, option, nat, list},
morekeywords=[4]{pose, set, move, case, elim, apply, clear, hnf,
  intro, intros, generalize, rename, pattern, after, destruct,
  induction, using, refine, inversion, injection, rewrite, congr,
  unlock, compute, ring, field, fourier, replace, fold, unfold,
  change, cutrewrite, simpl, have, suff, wlog, suffices, without,
  loss, nat_norm, assert, cut, trivial, revert, bool_congr, nat_congr,
  symmetry, transitivity, auto, eauto, eapply, split, left, right, autorewrite,
  by, done, exact, reflexivity, tauto, romega, omega,
  assumption, solve, contradiction, discriminate},
morekeywords=[5]{},
morekeywords=[6]{Var, App, Zero, Add, E, Dif, D},
morecomment=[s]{(*}{*)},
showstringspaces=false,
morestring=[b]",
morestring=[d],
tabsize=3,
extendedchars=false,
sensitive=true,
breaklines=false,
basicstyle=\footnotesize\ttfamily,
captionpos=b,
keepspaces=true,
identifierstyle={\ttfamily\color{black}},
keywordstyle=[1]{\ttfamily\color{purple}},
keywordstyle=[2]{\ttfamily\color{OliveGreen}},
keywordstyle=[3]{\ttfamily\color{ProcessBlue}},
keywordstyle=[4]{\ttfamily\color{RoyalBlue}},
keywordstyle=[5]{\ttfamily\color{Maroon}},
keywordstyle=[6]{\ttfamily\color{Bittersweet}},
stringstyle=\footnotesize\ttfamily,
commentstyle={\ttfamily\color{OliveGreen}},
literate=
    {\\forall}{{\color{dkgreen}{$\forall\;$}}}1
    {\\exists}{{$\exists\;$}}1
    {<-}{{$\leftarrow\;$}}1
    {=>}{{$\Rightarrow$}}1
    {==}{{\code{==}\;}}1
    {==>}{{\code{==>}\;}}1
    {->}{{$\rightarrow$}}1
    {<->}{{$\leftrightarrow\;$}}1
    {<==}{{$\leq\;$}}1
    {\#}{{$^\star$}}1
    {\\o}{{$\circ\;$}}1
    {\@}{{{\color{Bittersweet}$\cdot$}}}1
    {\/\\}{{$\wedge\;$}}1
    {\\\/}{{$\vee\;$}}1
    {++}{{\code{++}}}1
    {TILDE}{{\ }}1
    {\@\@}{{$@$}}1
    {\\mapsto}{{$\mapsto\;$}}1
    {\\hline}{{\rule{\linewidth}{0.5pt}}}1
    {LAM}{{{\color{Bittersweet}\textLambda}}}1
    {EPSILON}{{{\color{Bittersweet}\textepsilon}}}1
    {CDOT}{{{\color{Bittersweet}$\cdot$}}}1
    {+}{{{\color{Bittersweet}+}}}1
    {0}{{{\color{Bittersweet}0}}}1
    {TILDE}{{$\sim$}}1
    {\{}{{{\color{Bittersweet}\{}}}1
    {\}}{{{\color{Bittersweet}\}}}}1
    {|}{{{\color{Bittersweet}|}}}1
}[keywords,comments,strings]

\lstnewenvironment{coq}{\lstset{language=Coq}}{}
\def\coqe{\lstinline[language=Coq, basicstyle=\footnotesize\ttfamily]}

\begin{document}

\title[The Difference \texorpdfstring{$\lambda$}{lambda}-calculus]{The Difference \texorpdfstring{$\lambda$}{lambda}-calculus:\texorpdfstring{\\}{}
A Language for Difference Categories}

\titlecomment{This is an extended version of a paper appearing in
the 5th International Conference on Formal Structures for Computation and Deduction}

\author[M. Alvarez-Picallo]{Mario Alvarez-Picallo}
\address{Department of Computer Science, University of Oxford}
\curraddr{Edinburgh Research Center, Huawei Technologies Research \& Development}
\email{mario.alvarez-picallo@cs.ox.ac.uk}

\author[C.-H. L. Ong]{C.-H. Luke Ong}
\address{Department of Computer Science, University of Oxford}
\email{luke.ong@cs.ox.ac.uk}

\maketitle

\begin{abstract}
  Cartesian difference categories are a recent generalisation of Cartesian
  differential categories which introduce a notion of ``infinitesimal'' arrows
  satisfying an analogue of the Kock-Lawvere axiom, with the axioms of a
  Cartesian differential category being satisfied only ``up to an infinitesimal
  perturbation''. In this work, we construct a simply-typed calculus in the
  spirit of the differential $\lambda$-calculus equipped with syntactic
  ``infinitesimals'' and show how its models correspond to difference
  $\lambda$-categories, a family of Cartesian difference categories equipped
  with suitably well-behaved exponentials.
\end{abstract}

\section{Introduction}

A recent series of works introduced the concept of change actions and
differential maps between them \cite{alvarez2019fixing,alvarez2019change} in 
order to study settings equipped with derivative-like operations. Although the
motivating example was the eminently practical field of incremental computation,
similar structures appear in more abstract settings such as the calculus of 
finite differences and Cartesian differential categories. 

Of particular interest are Cartesian difference categories
\cite{alvarez2020cartesian}, a well-behaved class of change action models
\cite{alvarez2019change} that are much closer to the strong axioms of a
Cartesian differential category \cite{blute2009cartesian} while remaining
general enough for interpreting discrete calculus. A Cartesian difference
category is a left additive category equipped with an ``infinitesimal
extension'', an operation that sends an arrow $f$ to an arrow $\epsilon(f)$
which should be understood as $f$ being multiplied by an ``infinitesimal''
element -- infinitesimal in the sense that it verifies the Kock-Lawvere axiom
from synthetic differential geometry (we refer the reader to 
\cite{lavendhomme2013basic} for an introduction to SDG).

The interest of Cartesian differential categories is in part motivated by the
fact that they provide models for the differential $\lambda$-calculus
\cite{ehrhard2008differential,ehrhard2018introduction}, which extends the
$\lambda$-calculus with linear combinations of terms and an operator that
differentiates arbitrary $\lambda$-abstractions. The claim that differentiation
in the differential $\lambda$-calculus corresponds to the standard, ``analytic''
notion is then justified by its interpretation in (a well-behaved class of)
Cartesian differential categories
\cite{bucciarelli2010categorical,manzonetto2012categorical}.

It is reasonable to ask, then, whether there is a similar calculus that captures
the behavior of derivatives in difference categories -- especially since, as it
has been shown, these subsume differential categories. The issue is far from
trivial, as many of the properties of the differential $\lambda$-calculus
crucially hinge on derivatives being linear. Through this work we provide an
affirmative answer to this question by defining untyped and simply-typed
variants for a simple calculus which extends the differential $\lambda$-calculus
with a notion of derivative more suitable to the Cartesian difference setting.

\section{Cartesian Difference Categories}

The theory of Cartesian difference categories is developed and discussed at
length in \cite{alvarez2020cartesian,alvarez2020extended}, but we present here 
the main definitions and results which we will use throughout the paper,
referring the reader to \cite{alvarez2020extended} for the proofs.

\begin{defi}
  A \emph{Cartesian left additive category} (\cite[Definition
  1.1.1.]{blute2009cartesian}) $\cat{C}$ is a Cartesian category where every
  hom-set $\homset{\cat{C}}{A}{B}$ is endowed with the structure of a
  commutative monoid $(\homset{\cat{C}}{A}{B}, +, 0)$ such that $0 \circ f = 0$,
  $(f + g) \circ h = (f \circ h) + (g \circ h)$ and $\pair{f_1}{f_2} +
  \pair{g_1}{g_2} = \pair{f_1 + g_1}{f_2 + g_2}$.

  An \emph{infinitesimal extension} (\cite[Definition
  8]{alvarez2020cartesian}) in a Cartesian left additive category $\cat{C}$ is
  a choice of a monoid homomorphism $\epsilon : \homset{\cat{C}}{A}{B} \to
  \homset{\cat{C}}{A}{B}$ for every hom-set in $\cat{C}$. That is, $\varepsilon
  (f + g) = \epsilon(f) + \epsilon(g)$ and $\epsilon(0) = 0$. Furthermore, we
  require that $\epsilon$ be compatible with the Cartesian structure, in the
  sense that $
    \epsilon(\pair{f}{g}) = \pair{\epsilon(f)}{\epsilon(g)}
  $.
\end{defi}  

\begin{defi}
  A \emph{Cartesian difference category} (\cite[Definition
  9]{alvarez2020cartesian}) is a Cartesian left additive category
  with an infinitesimal extension $\varepsilon$ which is equipped with a
  \emph{difference combinator} $\dd{-}$ of the form: 
  \begin{gather*}
   \frac{f : A \to B}{\dd{f}: A \times A \to B}
  \end{gather*}
  satisfying the following coherence conditions (writing
  $\dd[^2]{f}$ for $\dd{\dd{f}}$):
  \begin{enumerate}[\CdCAx{\arabic*}\quad,ref={\CdCAx{\arabic*}},align=left]
  \setcounter{enumi}{-1}
    \item $f \circ (x + \varepsilon(u)) = f \circ x + \varepsilon\left( \dd{f}
    \circ \langle x, u \rangle \right)$ \label{CdC0}

    \item $\dd{f+g} = \dd{f} + \dd{g}$, $\dd{0} = 0$, and $\dd{\varepsilon(f)} =
    \varepsilon(\dd{f})$ \label{CdC1}

    \item $\dd{f} \circ \langle x, u + v \rangle = \dd{f} \circ \langle x, u
    \rangle + \dd{f} \circ \langle x + \varepsilon(u), v \rangle$ and $\dd{f}
    \circ \langle x, 0 \rangle = 0$\label{CdC2}

    \item $\dd{\Id{A}} = \pi_2$ and $\dd{\pi_1} = \pi_1 \circ \pi_2$ and $\dd{\pi_2} =
    \pi_2 \circ \pi_2$\label{CdC3}

    \item $\dd{\langle f, g \rangle} = \langle \dd{f}, \dd{g} \rangle$ and $\dd{!_A}
    = !_{A \times A}$ \label{CdC4}

    \item $\dd{g \circ f} = \dd{g} \circ \langle f \circ \pi_1, \dd{f} \rangle$
    \label{CdC5}

    \item $\dd[^2]{f} \circ \left \langle \langle  x, u \rangle,
    \langle 0, v \rangle \right \rangle= \dd{f} \circ  \langle x +
    \varepsilon(u), v \rangle$\label{CdC6}

    \item  $ \dd[^2]{f} \circ \left \langle \langle x, u \rangle,
    \langle v, 0 \rangle \right \rangle=  \dd[^2]{f} \circ \left
    \langle \langle x, v \rangle, \langle u, 0 \rangle \right \rangle$
    \label{CdC7}
  \end{enumerate}
\end{defi}
As noted in \cite{alvarez2020cartesian}, the axioms in a Cartesian differential
category (see e.g.~\CdCAx{1--7} in \ \cite{blute2009cartesian}) correspond to
the analogous axioms of the Cartesian difference operator, modulo certain
``infinitesimal'' terms, i.e.~terms of the form $\epsilon(f)$. We state here the
following two properties, whose proofs can be found in
\cite{alvarez2020extended}.
\begin{lem}
  \label{lem:d-epsilon}
  Given any map $f : A \to B$ in a Cartesian difference category $\cat{C}$, its
  derivative $\dd{f}$ satisfies the following equations:

  \begin{enumerate}[\roman{enumi}.,ref={\thelem.\roman{enumi}}]
    \item $\dd{f} \circ \pair{x}{\epsilon(u)} = \epsilon(\dd{f}) \circ
    \pair{x}{u}$
    \label{lem:d-epsilon-i}
    \item $\epsilon(\dd[^2]{f}) \circ \four{x}{u}{v}{0}
    = \epsilon^2(\dd[^2]{f}) \circ \four{x}{u}{v}{0}$
    \label{lem:d-epsilon-iii}
  \end{enumerate}
\end{lem}

\section{Difference \texorpdfstring{$\lambda$}{Lambda}-Categories}

In order to give a semantics for the differential $\lambda$-calculus, it does
not suffice to ask for a Cartesian differential category equipped with
exponentials -- the exponential structure has to be compatible with both the
additive and the differential structure, in the sense of
\cite[Definition~4.4]{bucciarelli2010categorical}. For difference categories we
will require an identical equation, together with a condition requiring
higher-order functions to respect the infinitesimal extension.

\begin{defi}
  \label{def:difference-lambda-category}
  We remind the reader that a Cartesian left additive category is
  \emph{Cartesian closed left additive}
  (\cite[Definition~4.2]{bucciarelli2010categorical}) whenever it is Cartesian
  closed and satisfies $\Lambda(f + g) = \Lambda(f) + \Lambda(g), \Lambda(0) = 0$.

  A Cartesian difference category $\cat{C}$ is a \emph{difference
  $\lambda$-category} if it Cartesian closed left additive and satisfies the
  following additional axioms:
  \begin{enumerate}[{\bf[$\partial\lambda$.\arabic*]},ref={\bf[$\partial\lambda$.{\arabic*}]},align=left]
    \item $\dd{\Lambda(f)} = \Lambda\pa{\dd{f}
      \circ \pair{\pa{\pi_1 \times \Id{}}}{\pa{\pi_2 \times 0}}
    }
    $
    \label{L1}
    \item $\Lambda(\epsilon(f)) = \varepsilon\pa{\Lambda(f)}$
    \label{L2}
  \end{enumerate}

  Equivalently, let $\A{sw}$ denote the map
  $
    \pair{\pair{\pi_{11}}{\pi_2}}{\pi_{21}}
    : (A \times B) \times C \to (A \times C) \times B
  $.
  Then the condition \ref{L1} can be written in terms of $\A{sw}$ as:
  \[
    \dd{\Lambda(f)} \defeq \Lambda\pa{
      \dd{f}
      \circ \pa{\Id{} \times \pair{\Id{}}{0}}
      \circ \A{sw}
    }
    \]
\end{defi}

Axiom \ref{L1} is identical to its differential analogue \cite[Definition
4.4]{bucciarelli2010categorical}, and it follows the same broad intuition. Given
a map $f : A \times B \to C$, we usually understand the composite $\dd{f} \circ
(\Id{A \times B} \times (\Id{A} \times 0_B)) : (A \times B) \times A \to C$ as a
partial derivative of $f$ with respect to its first argument. Hence, just as it
was with differential $\lambda$-categories, axiom~\ref{L1} states that the
derivative of a curried function is precisely the derivative of the uncurried
function with respect to its first argument.

\begin{exa}
  Let $\cat{C}$ be a differential $\lambda$-category. Then the trivial Cartesian
  difference category obtained by setting $\epsilon(f) = 0$ (as in
  \cite[Proposition 1]{alvarez2020cartesian}) is a difference
  $\lambda$-category. Furthermore, the Kleisli category $\cat{C}_\mathsf{T}$
  induced by its tangent bundle monad (as in \cite[Proposition
  6]{alvarez2020cartesian}) is also a difference $\lambda$-category.
\end{exa}

\begin{exa}
  The category $\overline{\cat{{Ab}}}$ (\cite[Section
  5.2]{alvarez2020cartesian}), which has Abelian groups as objects and arbitrary
  functions between their carrier sets as morphisms, is a difference
  $\lambda$-category with infinitesimal extension $\epsilon(f) = f$ and
  difference combinator $\dd{f}(x, u) = f(x + u) - f(x)$. Given
  groups $G, H$, the exponential $G \Ra H$ is the set of (set-theoretic)
  functions from $G$ into $H$, endowed with the group structure of $H$ lifted
  pointwise (that is, $(f + g)(x) = f(x) + g(x)$). Evidently the exponential
  respects the monoidal structure and the infinitesimal extension. We check that
  it also verifies axiom \ref{L1}:
  \begin{align*}
    \dd{\Lambda(f)}(x, u)(y) &= \Lambda(f)(x + u)(y) - \Lambda(f)(x)(y)
    \\
    &= f(x + u, y) - f(x, y)
    \\
    &= \Lambda(\dd{f} \circ (\Id{} \times \pair{\Id{}}{0}) \circ \A{sw})(x, u)(y)
  \end{align*}
\end{exa}

A central property of differential $\lambda$-categories is a deep correspondence
between differentiation and the evaluation map. As one would expect, the partial
derivative of the evaluation map gives one a first-class derivative operator
(see, for example, \cite[Lemma~4.5]{bucciarelli2010categorical}, which provides
an interpretation for the differential substitution operator in the differential
$\lambda$-calculus). This property still holds in difference categories,
although its formulation is somewhat more involved.

\begin{lem}
  \label{lem:lambda-d-ev}
  For any $\cat{C}$-morphisms $\Lambda(f) : A \to (B \Ra C), e : A \to B$, the
  following identities hold:
  \begin{enumerate}[\roman{enumi}.,ref={\thelem.\roman{enumi}}]
    \item $
    \dd{\A{ev} \circ \pair{\Lambda(f)}{e}}
    = \A{ev} \circ \pair{\dd{\Lambda(f)}}{e \circ \pi_1}
    +\dd{f} \circ \pair{\pair{\pi_1 + \epsilon(\pi_2)}{e \circ
    \pi_1}}{\pair{0}{\dd{e}}}$
    \item $
      \dd{\A{ev} \circ \pair{\Lambda(f)}{e}}
      = 
      \A{ev} \circ \pair{\dd{\Lambda(f)}}{e \circ \pi_1 + \epsilon(\dd{e})}
      +\dd{f} \circ \pair{\pair{\pi_1}{e \circ \pi_1}}{\pair{0}{\dd{e}}}
    $
  \end{enumerate}
\end{lem}

As is the case in differential $\lambda$-categories, we can define a
``differential substitution'' operator on the semantic side. This operator
is akin to post-composition with a partial derivative, and can be defined
as follows.

\begin{defi}
  \label{def:star-composition}
  Given morphisms $s : A \times B \ra C, u : A \ra B$, we define their
  \textbf{differential composition} $s \star u : A \times B \ra C$ by:
  \[
    s \star u \defeq \dd{s} \circ \pair{\Id{A \times B}}{\pair{0_A}{u \circ \pi_1}}
  \]
\end{defi}

We should understand the morphism $s \star u$ as the partial derivative of $s$
in its second argument, pre-composed with the morphism $u$. This operator verifies
the following useful properties (proofs of which can be found in \cite[Chapter 6.4]{alvarez2020change}).

\begin{lem}
  \label{lem:lambda-star-ev}
  Let $f : A \times B \times C \to D, g : A \to B, g' : A \times B \to B, e : A
  \times B \to C$ be arbitrary $\cat{C}$-morphisms. Then the following
  identities hold:
  \begin{enumerate}[\roman{enumi}.,ref={\thelem.\roman{enumi}}]
    \item {\small$\pa{\A{ev} \circ \pair{\Lambda(f)}{e}} \star g
      = \A{ev} \circ \pair{
        \Lambda(f \star (e \star g))
      }{e} 
      + \A{ev} 
      \circ \pair{
        \Lambda(f) \star g
      }{
        e \circ \pair{\pi_1}{\pi_2 + \epsilon(g) \circ \pi_1}
      }
    $}
    \item {\small $
      \Lambda (f \star e) \star g =
      \Lambda \Big[
        \Lambda^-(\Lambda(f) \star g) \star (e \circ (\Id{} + \pair{0}{\varepsilon(g)})))
      + \varepsilon (f \star e) \star (e \star g)
      + (f \star (e \star g))
      \Big]
      $}
    \item 
    {\small $\Lambda(f \star e) \circ \pair{\pi_1}{g'} 
      = \Lambda(\Lambda^-(\Lambda(f) \circ \pair{\pi_1}{g'})
      \star (e \circ \pair{\pi_1}{g'}))$}
  \end{enumerate}
\end{lem}

\newcommand{\wfParTo}[0]{\mathrel{\ul\parTo}}

\newcommand{\permEq}[0]{\sim_+}
\newcommand{\diffEq}[0]{\sim_\epsilon}
\newcommand{\diffEta}[0]{\sim_{\epsilon\eta}}
\newcommand{\canEq}[0]{\sim_\textbf{can}}
\newcommand{\wfReducesTo}[0]{\mathrel{\ul{\reducesTo}}}
\newcommand{\wfTs}[0]{\mathrel{\ul{\ts}}}
\newcommand{\ap}[2]{\textbf{ap}(#1, #2)}
\newcommand{\pri}[1]{\textbf{pri}(#1)}
\renewcommand{\tan}[1]{\textbf{tan}(#1)}
\newcommand{\can}[1]{\textbf{can}\left(#1\right)}
\newcommand{\reg}[1]{\textbf{reg}\left(#1\right)}

\section{An Untyped Calculus of Differences}
\label{sec:untyped}

We proceed in a manner similar to Vaux \cite{vaux2009algebraic} in his treatment
of the algebraic $\lambda$-calculus; that is, we will first define a set of
``unrestricted'' terms $\Lambda_\epsilon$ which we will later consider up to an
equivalence relation arising from the theory of difference categories.

\begin{defi} The set $\Lambda_\epsilon$ of \textbf{unrestricted terms} of
  the $\lambda_\epsilon$-calculus is given by the following inductive
  definition (assuming, as is usual, a countably infinite set of distinct variables 
  $x, y, z, \ldots$):
  \[
    \begin{array}{rccl}
      \text{Terms: } &s, t, e& \defeq & x
                                        \alter \lambda x . t
                                        \alter (s\ t)
                                        \alter \D(s) \cdot t
                                        \alter \epsilon t
                                        \alter s + t 
                                        \alter 0
    \end{array}
  \]
\end{defi}

Since the only binder in the $\lambda_\epsilon$-calculus
is the usual $\lambda$-abstraction, the definition of free and bound variables
is straightforward.

\begin{defi}
  The set of \textbf{free variables} $\text{FV}(t)$ of a term $t \in
  \Lambda_\epsilon$ is defined by induction on the structure of $t$ as follows:
  \[
    \begin{array}{rcl}
      \text{FV}(x) &\defeq& \lbrace x \rbrace\\
      \text{FV}(\lambda x . t) &\defeq& \text{FV}(t) \setminus \lbrace x \rbrace\\
      \text{FV}(s\ t) &\defeq& \text{FV}(s) \cup \text{FV}(t)\\
      \text{FV}(\D(s)\cdot t) &\defeq& \text{FV}(s) \cup \text{FV}(t)\\
      \text{FV}(\epsilon t) &\defeq& \text{FV}(t)\\
      \text{FV}(s + t) &\defeq& \text{FV}(s) \cup \text{FV}(t)\\
      \text{FV}(0) &\defeq& \emptyset\\
    \end{array}
  \]

  As usual, a variable $x$ is \textbf{free} in $t$ whenever $x \in \text{FV}(t)$.
  An occurrence of a variable $x$ in some term $t$ is said to be \textbf{bound}
  whenever it appears in some subterm $t'$ of $t$ with $x \not\in \text{FV}(t)$.
  Two terms are said to be $\alpha$-equivalent if they are identical up to a
  renaming of all their bound variables.
\end{defi}

In what follows, we will speak of terms only up to $\alpha$-equivalence. That
is, we consider the terms $\lambda x . x$ and $\lambda y . y$ to be identical
for all intents and purposes. Since this means we can rename bound variables
freely, we will assume by convention that all bound variables appearing in any
term $t \in \Lambda_\epsilon$ are different from its free variables.

\subsection{Differential Equivalence}

Further to $\alpha$-equivalence, we introduce here the notion of differential
equivalence of terms. The role of this relation is, as in
\cite{vaux2006lambda}, to enforce that the elementary algebraic properties of
sums and actions are preserved. For example, we wish to treat the terms $\lambda
x . (0 + \epsilon (s + t))$ and $(\lambda x . \epsilon t) + (\lambda x .
\epsilon s)$ as if they were equivalent (as it will be the case in the models).
This equivalence relation also has the role of ensuring that the axioms of a
Cartesian difference category are verified, especially regularity of
derivatives.

\begin{defi} A binary relation $\sim{} \subseteq \Lambda_\epsilon \times
  \Lambda_\epsilon$ is \textbf{contextual} whenever it satisfies the conditions
  in Figure~\ref{fig:contextual-relation} below.
  \begin{figure}[ht]
    \[
      \begin{array}{ccc}
        t \sim t' & \Rightarrow & \lambda x . t \sim \lambda x . t'\\
        t \sim t' & \Rightarrow & \epsilon t \sim \epsilon t'\\
        s \sim s' \wedge t \sim t' & \Rightarrow & s\ t \sim s'\ t'\\
        s \sim s' \wedge t \sim t' & \Rightarrow & \D(s)\cdot t \sim \D(s')\cdot t'\\
        s \sim s' \wedge t \sim t' & \Rightarrow & s + t \sim s' + t'\\
      \end{array}
    \]
    \caption{Contextuality on unrestricted $\Lambda_\epsilon$-terms}
    \label{fig:contextual-relation}
  \end{figure}
\end{defi}

\begin{lem}
  \label{lem:contextual-subterm}
  Whenever $\sim$ is contextual, if $t$ is a subterm of $s$ and $t \sim t'$ then
  $s \sim s'$, where $s'$ is the term resulting from substituting the occurrence
  of $t$ in $s$ for $t'$.
\end{lem}

\begin{defi} 
  \label{defn:differential-equivalence}
  \textbf{Differential equivalence} $\diffEq \subseteq \Lambda_\epsilon \times
  \Lambda_\epsilon$ is the least equivalence
  relation which is contextual and contains the relation $\diffEq^1$ below
  Figure~\ref{fig:differential-equivalence} below.
  \begin{figure}[ht]
  \[
    \begin{array}{rcl}
      (s + t) + e           &\diffEq^1& s + (t + e)\\
      s + 0                 &\diffEq^1& s\\
      s + t                 &\diffEq^1& t + s\\
      \epsilon 0            &\diffEq^1& 0\\
      \epsilon (s + t)      &\diffEq^1& \epsilon s + \epsilon t
    \end{array}
    \begin{array}{rcl}
      \lambda x . 0          &\diffEq^1& 0\\
      \lambda x . (s + t)    &\diffEq^1& (\lambda x . s) + (\lambda x . t)\\
      \lambda x . \epsilon t &\diffEq^1& \epsilon (\lambda x . t)\\
      0\ s                   &\diffEq^1& 0\\
      (s + t)\ e             &\diffEq^1& (s\ e) + (t\ e)\\
      (\epsilon s)\ t        &\diffEq^1& \epsilon (s\ t)
    \end{array}
  \]
  \[
    \begin{array}{rcl}
      \D(0) \cdot e          &\diffEq^1& 0\\
      \D(s + t) \cdot e      &\diffEq^1& (\D(s) \cdot e) + (\D(t) \cdot e)\\
      \D(\epsilon t) \cdot e &\diffEq^1& \epsilon (\D(t) \cdot e)
    \end{array}\]\[
    \begin{array}{rcl}
      \D(s) \cdot 0          &\diffEq^1& 0\\
      \D(s) \cdot (t + e)    &\diffEq^1& 
      \D(s) \cdot t
      + \D(s) \cdot e
      + \epsilon (\D(\D(s) \cdot t) \cdot e)\\
      \D(s) \cdot (\epsilon t) &\diffEq^1& \epsilon(\D(s) \cdot t)\\
      \D(\D(s) \cdot t) \cdot e &\diffEq^1& \D(\D(s) \cdot e) \cdot t\\
      \epsilon^2 \D(\D(s) \cdot t) \cdot e &\diffEq^1& \epsilon \D(\D(s) \cdot t) \cdot e\\
      s\ (t + \epsilon e) &\diffEq^1& (s\ t) + \epsilon((\D(s) \cdot e)\ t)
    \end{array}
  \]
  \caption{Differential equivalence on unrestricted $\Lambda_\epsilon$-terms}
  \label{fig:differential-equivalence}
  \end{figure}
\end{defi}

The above conditions can be separated in a number of conceptually distinct
groups corresponding to their purpose. These are as follows:
\begin{itemize}
  \item The first block of equations simply states that $+, 0$ define a
    commutative monoid and that $\epsilon$ defines a monoid homomorphism.
  \item The second block of equations amounts to stating that the monoid and
    infinitesimal extension structure on functions is pointwise.
  \item The third block of equations implies (and is equivalent to) stating that
    addition and infinitesimal extension are ``linear'', in the sense that they
    are equal to their own derivatives (that is, $\dd{+} = + \circ \pi_2$ and
    $\dd{\epsilon}= \epsilon$).
  \item The fourth block of equations states structural properties of the
    derivative, such as the derivative conditions and the commutativity of second
    derivatives. Similar equations are also present in the differential
    $\lambda$-calculus, where they merely state that the derivative is additive
    (as opposed to regular).
\end{itemize}

Most of these equations correspond directly to properties of Cartesian
difference categories, with the only exception being the requirement that $\D(s)
\cdot (\epsilon t) \diffEq \epsilon (\D(s) \cdot t)$ and the ``duplication
of infinitesimals'' in $\epsilon \D (\D(s) \cdot t) \cdot e = \epsilon^2
\D(\D(s) \cdot t) \cdot e$. To understand the logic of the first equation, consider
that the term $\D(\lambda x . s) \cdot \epsilon t$ should roughly correspond to
$\dd{\sem{s}}(x, \epsilon \sem{t})$ which expands to:

\begin{align*}
  \dd{\sem{s}}(x, \epsilon {\sem{t}}) &= \dd{\sem{s}}(x, 0 + \epsilon {\sem{t}})\\
  &= \dd{\sem{s}}(x, 0) + \epsilon \dd[^2]{\sem{s}}(x, 0, 0, {\sem{t}})\\
  &= \epsilon \dd[^2]{\sem{s}}(x, 0, 0, {\sem{t}})
\end{align*}
In the particular setting of Cartesian difference categories, the last term
boils down to $\epsilon \dd{\sem{s}}(x, \sem{t})$ by axiom~\ref{CdC6}, hence the
equality $\D(s) \cdot \epsilon t \diffEq \epsilon (\D(s) \cdot t)$ is justified,
allowing infinitesimal extensions to be ``pulled out'' of differential
application. The duplication of infinitesimals is simply a syntactic version of
Lemma~\ref{lem:d-epsilon-iii}.

\begin{defi}
  The set $\lambda_\epsilon$ of \textbf{well-formed terms}, or simply
  \textbf{terms}, of the $\lambda_\epsilon$-calculus is defined as the quotient
  set $\lambda_\epsilon \defeq \Lambda_\epsilon/\diffEq$. Whenever $t$ is an
  unrestricted term, we write $\ul{t}$ to refer to the well-formed term
  represented by $t$, that is to say, the $\diffEq$-equivalence class of $t$. 
\end{defi}

The notion of differential equivalence allows us to ensure that our calculus
reflects the laws of the underlying models, but has the unintended
consequence that our $\lambda_\epsilon$-terms are equivalence classes, rather
than purely syntactic objects. We will proceed by defining a notion of canonical
form of a term and a canonicalization algorithm which explicitly constructs
the canonical form of any given term, thus proving that $\diffEq$ is decidable.

\newcommand{\canonical}[0]{\text{C}_\epsilon}
\newcommand{\base}[0]{\text{B}_\epsilon}
\newcommand{\additive}[0]{\base^*}
\newcommand{\positive}[0]{\base^+}
\newcommand{\posCanonical}[0]{\canonical^+}
\renewcommand{\b}[0]{\textbf{b}}

\begin{defi}
  We define the sets $\base \subset \positive \subset \additive \subset
  \posCanonical \subset \canonical (\subset \Lambda_\epsilon)$ of \textbf{basic},
  \textbf{positive}, \textbf{additive}, \textbf{positive canonical} and
  \textbf{canonical} terms according to the following grammar:
\[\def\arraystretch{1.3}
  \begin{array}{rccl}
    \text{Basic terms: } &s^\b, t^\b, e^\b \in \base& \defeq & x
                                                  \alter \lambda x . t^\b
                                                  \alter (s^\b\ t^*)
                                                  \alter \D(s^\b) \cdot t^\b\\
    \text{Positive terms: } &s^+, t^+, e^+ \in \positive& \defeq & s^\b \alter s^\b + (t^+)\\
    \text{Additive terms: } &s^*, t^*, e^* \in \additive& \defeq & 0 \alter s^+\\
    \text{Positive canonical terms: } &S^+, T^+ \in 
      \canonical&\defeq& \epsilon^k s^\b \alter \epsilon^k s^\b + (S^+) \\
    \text{Canonical terms: } &S, T \in \canonical&\defeq& 0 \alter S^+
  \end{array}
\] We will sometimes abuse the notation and write $\ul{t^*}$ or $\ul{ t^\b }$ to
  denote well-formed terms whose canonical form is an additive or basic term
  respectively.
\end{defi}

The above definition is somewhat technical, so a more informal description of
canonical forms should be helpful. We observe that all the rules of differential
equivalence can be oriented from left to right to obtain a rewrite system -
intuitively, this rewrite system operates by pulling all the instances of
addition and infinitesimal extension to the outermost layers of a term. Since
every syntactic construct is additive except for application, basic terms may only
contain additive terms as the arguments to a function application. 

As infinitesimal extensions are themselves additive, we also want to disallow
terms such as $\epsilon (s + t)$, instead factoring out the extension
into $\epsilon s + \epsilon t)$. A general canonical term $T \in \canonical$ then has
the form:
    \[
  T = \epsilon^{k_1} t^\b_1 + (\epsilon^{k_2}t_2 + (\ldots + \epsilon^{k_n} t^\b_n) \ldots )
\]
That is to say, a canonical term is similar to a polynomial with coefficients in
the set of basic terms and a variable $\epsilon$ (but note that canonical terms
are always written in their ``fully distributed'' form, that is, we write
$\epsilon s + (\epsilon t + \epsilon^2 e)$ rather than $\epsilon ((s + t) + \epsilon
e)$).

We will freely abuse the notation and write $\sum_{i=1}^n \epsilon^{k_i} t^\b_i$ to
denote a general canonical term, as this form is easier to manipulate in many
cases. In particular, the canonical term $0$ is precisely the sum of zero terms.

\begin{defi}
  Given unrestricted terms $s, t$, we define their \textbf{canonical sum} $s
  \cansum t$ by induction as follows:
  \begin{align*}
    0 \cansum t &\defeq t\\
    s \cansum 0 &\defeq s\\
    (s + s') \cansum t &\defeq s + (s' \cansum t)
  \end{align*}
\end{defi}

\begin{lem}
  The canonical sum $S \cansum T$ of any two canonical terms is a canonical
  term. Furthermore $\cansum$ is associative and has $0$ as an identity
  element. When $S_i$ are canonical terms, we will write $\boxsum_{i=1}^n S_i$
  for the canonical term $S_1 \cansum (S_2 \cansum \ldots S_n)\ldots)$.
\end{lem}

\begin{defi} Given an unrestricted $\lambda_\epsilon$-term $t \in
  \Lambda_\epsilon$, we define its \textbf{canonical form} $\can{t}$ by structural
  induction on $t$ as follows:
  \begin{itemize}
    \item $\can{0} \defeq 0$
    \item $\can{x} \defeq x$
    \item $\can{s + t} \defeq \can{s} \cansum \can{t}$
    \item $\can{\epsilon t} \defeq \epsilon^* \can{t}$, where 
      \[
        \epsilon^* T \defeq \begin{cases}
          \sum_{i=1}^n \epsilon^*\epsilon^{k_i}t^\b_i
          &\text{ if $T = \sum_{i=1}^n \epsilon^{k_i}t^\b_i$}\\
          T 
          &\text{ if $T = \epsilon \D(\D(e) \cdot u) \cdot v$}\\
          \epsilon T &\text{ otherwise}
        \end{cases}
      \]
    \item If 
      $\can{t} = \sum_{i=1}^{n}\epsilon^{k_i}t^\b_{i}$ 
      then
      \[ \can{\lambda x.t} \defeq \sum_{i=1}^n \epsilon^{k_i} (\lambda^* x.t^\b_i) \]
    \item If
      $\can{s} = \sum_{i=1}^n\epsilon^{k_i}s^\b_i$
      and
      $\can{t} = T$
      then
      \[
        \can{\D(s) \cdot t} \defeq \boxsum_{i=1}^n ((\epsilon^*)^{k_i} \reg{s^\b_i, T})
      \]
      where the \textbf{regularization} $\reg{s, T}$ is
      defined by structural induction on $T$:
      \begin{align*}
        \reg{s, 0} &\defeq 0\\
        \reg{s, \epsilon^{k} t^{\b} + T'} &\defeq
        \brak{\pa{\epsilon^*}^{k} \D\pa{s}\cdot t^\b}
        \cansum \brak{
          \reg{s, T'}
        }\\
        \continued \cansum \brak{
          \pa{\epsilon^*}^{k+1} \D^*\pa{\reg{s, T'}} \cdot t^\b
        }
      \end{align*}
      and $\D^*$ denotes the extension of $\D$ by additivity in its first
      argument, that is to say:
      \begin{align*}
        \D^*\pa{\sum_{i=1}^n \epsilon^{k_i} s^\b_i} \cdot t^\b
        &\defeq \sum_{i=1}^n \epsilon^{k_i} \pa{s^\b_i\ t^\b}
      \end{align*}
      Observe that, whenever $S$ is canonical and $t^\b$ is basic, the term
      $\D^*(S) \cdot t^\b$ is also canonical. Therefore, by induction, the
      regularization $\reg{s^\b, T}$ is indeed a canonical term term, since
      canonicity is preserved by $\epsilon^*, \cansum$.
    \item
      If $\can{s} = \sum_{i=1}^n \epsilon^{k_i} s^\b_i$ and $\can{t} = T$, then
      \begin{align*}
        \can{s\ t} &\defeq \brak{\sum_{i=1}^n \epsilon^{k_i} \pa{s^\b_i \ \textbf{pri}(T)}}
        \cansum
        \brak{
          \epsilon^* \pa{
            \sum_{i=1}^n \textbf{ap}(\reg{s^\b_i, \textbf{tan}(T)}, \textbf{pri}(T))
          }
        }
      \end{align*}
      where the primal $\textbf{pri}$ and tangent $\textbf{tan}$ components of a
      canonical term $T$ correspond respectively to the basic terms with zero and
      non-zero $\epsilon$ coefficients, and $\textbf{ap}$ is the additive
      extension of application.
      \[\def\arraystretch{1.3}
      \begin{array}{rclrcl}
        \textbf{pri}\pa{0} &\defeq& 0 
        &\textbf{tan}\pa{0} &\defeq& 0\\
        \textbf{pri}\pa{\epsilon^{k+1}t^\b + T'} &\defeq& \textbf{pri}\pa{T'}
                                               &\textbf{tan}\pa{\epsilon^{k+1}t^\b + T'} &\defeq& \epsilon^{k+1} t^\b + \textbf{tan}\pa{T'}\\
        \textbf{pri}\pa{\epsilon^{0}t^\b + T'} &\defeq& t^\b + \textbf{pri}\pa{T'}
        &\textbf{tan}\pa{\epsilon^{0}t^\b + T'} &\defeq& \textbf{tan}\pa{T'}\\
      \end{array}
      \]
      \[\def\arraystretch{1.3}
      \begin{array}{rcl}
        \textbf{ap}\pa{\sum_{i=1}^n \epsilon^{k_i} s^\b_i, t^\b} &\defeq&
        \sum_{i=1}^n \epsilon^{k_i} (s^\b_i\ t^\b)
      \end{array}
      \]
  \end{itemize}
\end{defi}

\begin{exa}
  The canonicalization algorithm is mostly straightforward, with only the cases
  for differential and standard application being of any interest. Since this is
  the part of the system where we diverge most from the differential
  $\lambda$-calculus, it is worth examining an example of regularization.
  Consider the following unrestricted term on free variables $u, x, y, z$:
  \[
    t \defeq \D (u) \cdot (x + y + \epsilon z)
  \]
  Applying the algorithm above, we compute its canonical form to be:
  \begin{align*}
    &\hspace{-10pt}\can{\D(u) \cdot (x + y + \epsilon z)}\\
    &= \reg{u, x + y + \epsilon z}\\
    &= 
      (\D(u) \cdot x) 
      \boxplus 
      \reg{u, y + \epsilon z} 
      \boxplus 
      (\epsilon^* \D^*(\reg{u, y + \epsilon z}) \cdot x )\displaybreak[0]\\
    &= 
      (\D(u) \cdot x) \boxplus
      \Big[ (\D(u) \cdot y) \boxplus \reg{u, \epsilon z} \boxplus \epsilon^* (\D^*(\reg{u, \epsilon z}) \cdot y) \Big]
      \\
      \continued \boxplus 
      (\epsilon^* \D^*(\reg{u, y + \epsilon z}) \cdot x )\displaybreak[0]\\
    &= (\D(u) \cdot x) \boxplus 
      \Big[ (\D(u) \cdot y) \boxplus (\epsilon \D(u) \cdot z) \boxplus \epsilon^*(\D^*
      (\epsilon \D(u) \cdot z) \cdot y))\Big]\\
      \continued
      \boxplus (\epsilon^* \D^*(\reg{u, y + \epsilon z}) \cdot x)\displaybreak[0]\\
    &= (\D(u) \cdot x)) \boxplus \Big[ \D(u) \cdot y + \epsilon (\D(u) \cdot z) + \epsilon(\D
      (\D(u) \cdot z) \cdot y) \Big]\\
      \continued
      \boxplus (\epsilon^* \D^*(\reg{u, y + \epsilon z}) \cdot x)
    \displaybreak[0]\\
    &= \Big[ \D(u) \cdot x + \D(u) \cdot y + \epsilon \pa{ \D(u) \cdot z } + \epsilon\pa{\D
      (\D(u) \cdot z) \cdot y)} \Big]\\
      \continued \boxplus \pa{ 
        \epsilon \pa{
          \D(\D(u) \cdot y) \cdot x
          + \epsilon \pa{
            \D(\D(u) \cdot z) \cdot x
            + \epsilon \pa{
              \D(\D(\D(u) \cdot z) \cdot y) \cdot x
            }
          }
        }
      }
    \displaybreak[0]\\
    &= \D(u) \cdot y + \D(u) \cdot x\\
    \continued + \epsilon \Big[
      \D(u) \cdot z + \D( \D(u) \cdot y) \cdot x\\
      \continued[60pt] + \epsilon (
        \D(\D(u) \cdot z) \cdot y + {}
        \D(\D(u) \cdot z) \cdot x +
        \epsilon \left( \D(\D(\D(u) \cdot z) \cdot y ) \cdot x\right))\Big]
  \end{align*}
  This is precisely the result we would expect from fully unfolding the expression 
  $\dd{u}(a, x + y + \epsilon(z))$ in a Cartesian difference category and repeatedly
  applying regularity of the derivative!
\end{exa}

\begin{thm}
  \label{thm:canonical-form-differential}
  Every unrestricted $\lambda_\epsilon$-term is differentially equivalent to its
  canonical form. That is to say, for all $t \in \Lambda_\epsilon$, we have $t
  \diffEq \can{t}$.
\end{thm}
\begin{proof}
  The proof proceeds by straightforward induction on $t$. We explicitly prove
  the case for the canonicalization of differential and standard application, as
  these are the only nontrivial cases. For this we will make use of the
  following results:

  \begin{lem}
    \label{lem:utils-differential}
    Given canonical terms $S, T$, we have:
    \begin{align*}
      \epsilon^* S &\diffEq \epsilon S\\
      S \cansum T &\diffEq S + T\\
      \D^* (S) \cdot T &\diffEq \D(S) \cdot T\\
      \textbf{ap}\pa{S, T} &\diffEq S\ T
    \end{align*}
  \end{lem}
  \begin{proof}
    All four results follow by straightforward structural induction on $S$.
  \end{proof}
  \begin{lem}
    \label{lem:reg-differential}
    Given a basic term $s^\b$ and a canonical term $T$ the differential
    application $\D(s^\b) \cdot T$ is differentially equivalent to its
    regularized version $\reg{s^\b, T}$.
  \end{lem}
  \begin{proof}
    The proof follows by induction on the structure of $T$.
    \begin{itemize}
      \item When $T = 0$ we have
        $\reg{s^\b, 0} = 0 \diffEq \D(s^*) \cdot 0$
      \item When $T = \epsilon^k t^\b + T'$ we have:
      \begin{align*}
        \reg{s^\b, \epsilon^k t^\b + T'} &= \brak{(\epsilon^*)^k \D(s^\b) \cdot t^\b} \cansum \brak{\reg{s^\b, T'}}\\
        \continued \cansum \brak{(\epsilon^*)^{k+1} \D^*(\reg{s^\b, T'}) \cdot t^\b}\\
        &\diffEq \epsilon^k \D(s^\b) \cdot t^\b + \D(s^\b) \cdot T' + \epsilon^{k+1} \D(\D(s^\b) \cdot T') \cdot t^\b\qedhere
      \end{align*}
    \end{itemize}
    \renewcommand{\qedsymbol}[0]{}
  \end{proof}

  Going back to the proof of Theorem~\ref{thm:canonical-form-differential}, the
  case for differential application is obtained as a straightforward corollary
  of Lemma~\ref{lem:reg-differential}. For conventional application, consider
  terms $s, t$, and note that if $\can{t} = T$ then $t \diffEq \textbf{pri}(T) +
  \epsilon \textbf{tan}(T)$. Then, for any basic term $s^\b$, we obtain:
  \begin{align*}
    \can{s^\b\ t} &= \pa{s^\b\ \textbf{pri}(T)} 
    \cansum \brak{
      \epsilon^* \textbf{ap}\pa{\reg{s^\b, \textbf{tan}(T)}, \textbf{pri}(T)}
    }\\
    &\diffEq \pa{s^\b\ \textbf{pri}(T)} 
    + \epsilon \brak{(\D(s^\b) \cdot \textbf{tan}(T))\ \textbf{pri}(T)}\\
    &\diffEq s^\b\ \brak{
      \textbf{pri}(T) + \epsilon \textbf{tan}(T)
    }\\
    &\diffEq s^\b\ t\qedhere
  \end{align*}
\end{proof}

Our canonicalization algorithm is a result of orienting the equations in
Figure~\ref{fig:differential-equivalence}. Note, however, that while most of
these equivalences have a ``natural'' orientation to them, two of them are
entirely symmetrical: those being commutativity of the sum and the derivative.
Barring the imposition of some arbitrary total ordering on terms which would
allow us to prefer the term $x + y$ over $y + x$ (or vice versa), we settle for
our canonical forms to be unique ``up to'' these commutativity conditions.

\begin{defi}
  \label{defn:permutative-equivalence}
  \textbf{Permutative equivalence} $\permEq{} \subseteq \Lambda_\epsilon \times
  \Lambda_\epsilon$ is the least equivalence relation which is contextual and
  satisfies the properties in Figure~\ref{fig:permutative-equivalence} below.
  \begin{figure}[ht]
    \[
      \begin{array}{rcl}
        s + (t + e) &\permEq{} & (s + t) + e\\
        s + t &\permEq{}& t + s\\
        \D(\D(s) \cdot t) \cdot e &\permEq{}& \D(\D(s) \cdot e) \cdot t
      \end{array}
    \]
    \caption{Permutative equivalence on unrestricted $\Lambda_\epsilon$-terms}
    \label{fig:permutative-equivalence}
  \end{figure}
\end{defi}

We need to include associativity in the definition of permutative
equality, as otherwise the canonical term $x + (y + z)$ would not be
permutatively equivalent to $y + (x + z)$.

\begin{thm}
  \label{thm:canonicity}
  Given unrestricted terms $s, t \in \Lambda_\epsilon$, they are differentially
  equivalent if and only if their canonical forms are permutatively equivalent.
  More succinctly, $s \diffEq t$ if and only if $\can{s} \permEq \can{t}$
\end{thm}
\begin{proof}
  As an immediate consequence of Theorem~\ref{thm:canonical-form-differential},
  we know that $s \diffEq t$ if and only if $\can{s} \diffEq \can{t}$. The
  desired result will then follow from the fact that differential equivalence is
  precisely equivalent to permutative equivalence on canonical terms.
  
  Before proving this, we remark that $\canEq$ is reflexive, transitive and
  symmetric, which follows immediately from its definition and reflexivity,
  transitivity and symmetry of $\permEq$. We also prove the following two
  helpful results:

  \begin{lem}
    \label{lem:can-contextual}
    For any unrestricted terms $s, t, e$, the following equalities hold:
    \begin{align*}
      \can{\can{s} + \can{t}} &= \can{s + t}\\
      \can{\epsilon\can{x}} &= \can{\epsilon(x)}\\
      \can{\can{s}\ \can{t}} &= \can{s\ t}\\
      \can{\D(\can{s}) \cdot \can{t}} &= \can{\D(s) \cdot t}
    \end{align*}
    As a consequence, the relation $\canEq$ is contextual.
  \end{lem}
  \begin{proof}
    Follows immediately from the definition of $\can{}$ and observing that it
    only depends on the canonicalization of the subterms of the outermost
    syntactic form.
  \end{proof}

  \begin{lem}
    \label{lem:can-permutative}
    Whenever $S, T$ are canonical terms, $S \diffEq T$ if and only if $S \permEq T$.
  \end{lem}
  \begin{proof}
    Since $\diffEq$ is the least reflexive, transitive, symmetric and
    contextual relation that verifies the conditions in
    Figure~\ref{fig:differential-equivalence}, it follows that whenever $S
    \diffEq T$ it must be the case that $S = S_1 \diffEq^\star S_2 \diffEq^\star
    \ldots \diffEq^\star S_n = T$, where $\diffEq^\star$ denotes the contextual,
    symmetric (but not transitive) closure of $\diffEq^1$.

    But for each such condition, we have $\can{\text{LHS}} =
    \can{\text{RHS}}$, except for the two commutativity conditions, and in all
    cases we have $\can{\text{LHS}} \permEq \can{\text{RHS}}$.
    Since $\canEq$ is contextual, we have that $S = S_1 \canEq S_2 \canEq \ldots
    \canEq S_n = T$, and thus by transitivity we obtain $S \canEq T$.
  \end{proof}

  Theorem~\ref{thm:canonicity} is then an immediate consequence of
  Lemma~\ref{lem:can-permutative} and
  Theorem~\ref{thm:canonical-form-differential}.
\end{proof}

\begin{cor}
  Differential equivalence of $\Lambda_\epsilon$-terms is decidable.
\end{cor}
\begin{proof}
  Permutative equivalence of two terms is decidable, since the set of
  $\permEq$-equivalence classes of any term is finite and can be enumerated
  easily. Canonicalization is also decidable, since it is defined as a clearly
  well-founded recursion.
\end{proof}

\begin{cor}
  The set $\lambda_\epsilon$ of well-formed terms corresponds precisely to the
  set of canonical terms up to permutative equivalence $\canonical / \permEq$.
\end{cor}

\subsection{Substitution}

As is usual, our calculus features two different kinds of application: standard
function application, represented as $(s\ t)$; and differential application,
represented as $\D(s) \cdot t$. These two give rise to two different notions of
substitution. The first is, of course, the usual capture-avoiding substitution.
The second, differential substitution, is similar to the equivalent notion in
the differential $\lambda$-calculus, as it arises from the same chain rule that
is verified in both Cartesian differential categories and change action models.

\begin{defi}
  Given terms $s, t \in \Lambda_\epsilon$ and a variable $x$, the
  \textbf{capture-avoiding substitution} of $s$ for $x$ in $t$ (which we write
  as $\subst{t}{x}{s}$) is defined by induction on the structure of $t$ as in
  Figure~\ref{fig:substitution} below.
  \begin{figure}[ht]
    \begin{center}
      \[
      \begin{array}{cccl}
        \subst{x}{x}{s} &\defeq& s\\
        \subst{y}{x}{s} &\defeq& y &\text{ if $x \neq y$}\\
        \subst{(\lambda y . t)}{x}{s} &\defeq& \lambda y . (\subst{t}{x}{s}) 
                        &\text{ if $y \not\in \text{FV}(s)$}\\
        \subst{(t\ e)}{x}{s} &\defeq& (\subst{t}{x}{s})(\subst{e}{x}{s})\\
        \subst{(\D(t) \cdot e)}{x}{s} &\defeq& \D(\subst{t}{x}{s}) \cdot (\subst{e}{x}{s}) \\
        \subst{(\epsilon t)}{x}{s} &\defeq& \epsilon (\subst{t}{x}{s})\\
        \subst{(t + e)}{x}{s} &\defeq& (\subst{t}{x}{s}) + (\subst{e}{x}{s})\\
        \subst{0}{x}{s} &\defeq& 0
      \end{array}
      \]
      \caption{Capture-avoiding substitution in $\lambda_\epsilon$}
      \label{fig:substitution}
    \end{center}
  \end{figure}
\end{defi}

\begin{prop}
  \label{prop:substitution-differential}
  Capture-avoiding substitution respects differential equivalence. That is to
  say, whenever $s \diffEq s'$ and $t \diffEq t'$, it is the case that
  $\subst{t}{x}{s} \diffEq \subst{t'}{x}{s'}$.
\end{prop}
\begin{proof}
  It suffices to show that $\subst{t}{x}{s} \diffEq^1 \subst{t'}{x}{s}$ (the
  full result will then follow from transitivity and contextuality of
  $\diffEq$), which can be proven by straightforward structural induction on $t$.
\end{proof}

\begin{defi}
  Given terms $s, t \in \Lambda_\epsilon$ and a variable $x$ \emph{which is not
  free in $s$}\footnote{One should emphasise this constraint. Differential
  substitution appears in the reduction of $\D(\lambda x . t) \cdot s$ into
  $\lambda x . \diffS{t}{x}{s}$, and so if $x$ were free in $s$ it would become
  bound by the enclosing $\lambda$-abstraction.}, the \textbf{differential
  substitution} of $s$ for $x$ in $t$, which we write as $\diffS{t}{x}{s}$, is
  defined by induction on the structure of $t$ as in
  Figure~\ref{fig:differential-substitution}.
  \begin{figure}[t]
    \begin{center}
    \[\def\arraystretch{1.7}
      \begin{array}{cccl}
        \diffS{x}{x}{s} &\defeq& s\\
        \diffS{y}{x}{s} &\defeq& 0 &\text{ if $x \neq y$}\\
        \diffS{(\lambda y . t)}{x}{s} &\defeq& \lambda y . \pa{\diffS{t}{x}{s}}
                        &\text{ if $y \not\in \text{FV}(t)$}\\
        \diffS{(t\ e)}{x}{s} &\defeq& \brak{\D(t) \cdot \pa{\diffS{e}{x}{s}}\ e} +
                                      \brak{\diffS{t}{x}{s}\ \pa{\subst{e}{x}{(x
                                      + \epsilon s)}}}\\
        \diffS{(\D(t) \cdot e)}{x}{s} &\defeq& 
          \D(t) \cdot \pa{\diffS{e}{x}{s}} 
          + \D\pa{\diffS{t}{x}{s}} \cdot (\subst{e}{x}{(x + \epsilon s)})\\
        &&
          {}+ \epsilon \D(\D(t) \cdot e ) \cdot \pa{\diffS{e}{x}{s}}\\
        \diffS{(\epsilon t)}{x}{s} &\defeq& \epsilon \pa{\diffS{t}{x}{s}}\\
        \diffS{(t + e)}{x}{s} &\defeq& \pa{\diffS{t}{x}{s}} + \pa{\diffS{e}{x}{s}}\\
        \diffS{0}{x}{s} &\defeq& 0
      \end{array}
    \]
    \caption{Differential substitution in $\lambda_\epsilon$}
    \label{fig:differential-substitution}
    \end{center}
  \end{figure}
  We write 
  \[ 
    \frac{\partial^k t}{\partial (x_1, \ldots, x_k)}(u_1, \ldots, u_k)
   \]
  to denote the sequence of nested differential substitutions
  \[ 
    (\partial (\ldots ((\partial t / \partial x_1)(u_1)) \ldots) / \partial x_k)(u_k)
  \]
\end{defi}

Most of the cases of differential substitution are identical to those in the
differential $\lambda$-calculus (compare the above with the rules in e.g. \cite{ehrhard2003differential}). There are, however, a
number of notable differences which stem from our more general setting. First,
we must point out that this definition in fact coincides exactly with the
original notion of differential substitution in e.g
\cite{ehrhard2003differential}, provided that one assumes the identity $\epsilon
t = 0$ for all terms. This reflects the fact that every Cartesian differential
category is in fact a Cartesian difference category with trivial infinitesimal
extension.

All the differences in this definition stem from the failure of derivatives to
be additive in the setting of Cartesian difference categories. Consider the case
for $\diffS{\D (t) \cdot s}{x}{e}$, and remember that the ``essence'' of a
derivative in our setting lies in the derivative condition, that is to say, if
$t(x)$ is a term with a free variable $x$, we seek our notion of differential
substitution to satisfy a condition akin to Taylor's formula:
\[
  t(x + \epsilon y) \diffEq t(x) + \epsilon \diffS{t}{x}{y}
\]

When the term $t$ is a differential application, and assuming the above
``Taylor's formula'' holds for all of its subterms (which we will show later),
this leads us to the following informal argument:
\begin{align*}
  \D(t(x + \epsilon y)) \cdot (s(x + \epsilon y))
  &\diffEq \D\pa{t(x) + \epsilon \diffS{t}{x}{y}}  \cdot (s(x + \epsilon y))
  \displaybreak[0]\\
  &\diffEq \D(t(x)) \cdot (s(x + \epsilon y)) 
    + \epsilon \D\pa{\diffS{t}{x}{y}} \cdot (s (x + \epsilon y))
  \displaybreak[0]\\
  &\diffEq \D(t(x)) \cdot (s(x))\\
    &\qquad + \epsilon \D(t(x)) \cdot \pa{\diffS{s}{x}{y}}\\
    &\qquad + \epsilon \D\pa{\diffS{t}{x}{y}} \cdot (s (x + \epsilon y))\\
    &\qquad + \epsilon^2\D(\D(t(x)) \cdot (s(x))) \cdot \pa{\diffS{s}{x}{y}}
\end{align*}

From this calculation, the differential substitution for this case arises
naturally as it results from factoring out the $\epsilon$ and noticing that
the resulting expression has precisely the correct shape to be Taylor's formula
for the case of differential application. The case for standard application can
be derived similarly, although the involved terms are simpler. Differential
substitution verifies some useful properties, which we state below (mechanised
proofs are available in \cite[Appendix A]{alvarez2020change}, although the details 
are more cumbersome than enlightening).

\begin{prop}
  \label{prop:differential-substitution-differential}
  Differential substitution respects differential equivalence. That is to
  say, whenever $s \diffEq s'$ and $t \diffEq t'$, it is the case that
  $\diffS{t}{x}{s} \diffEq \diffS{t'}{x}{s'}$.
\end{prop}
\begin{proof}
  See \cite[\coqe{Lemma dsubst_diff}]{alvarez2020change}
\end{proof}
\begin{prop}
  \label{prop:differential-substitution-closed}
  Whenever $x$ is not free in $t$, then $\diffS{t}{x}{u} \diffEq 0$.
\end{prop}
\begin{proof}
  See \cite[\coqe{Lemma dsubst_empty}]{alvarez2020change}
\end{proof}
\begin{prop}
  \label{prop:differential-substitution-commute}
  Whenever $x$ is not free in $u, v$, then:
  \[
    \frac{\partial^2{t}}{\partial{x^2}}(u, v)
    \diffEq
    \frac{\partial^2 t}{\partial x^2}(v, u)
  \]
\end{prop}
\begin{proof}
  See \cite[\coqe{Lemma dsubst_commute}]{alvarez2020change}
\end{proof}

As we have previously mentioned, the rationale behind our specific definition of
differential substitution is that it should verify some sort of ``Taylor's
formula'' (or rather, Kock-Lawvere formula), in the following sense:
\begin{thm}
  \label{thm:taylor}
  For any unrestricted terms $s, t, e$ and any variable $x$ which does not
  appear free in $e$, we have
  \[
    \subst{s}{x}{t + \epsilon e} \diffEq \subst{s}{x}{t} 
    + \epsilon\pa{\subst{\pa{\diffS{s}{x}{e}}}{x}{t}}
  \]
  We will often refer to the right-hand side of the above equivalence as the
  \textbf{Taylor expansion} of the corresponding term in the left-hand side.
\end{thm}
\begin{proof}
  \newcommand{\lhs}[1]{\subst{#1}{x}{t + \epsilon e}}
  \newcommand{\diff}[1]{\subst{\pa{\diffS{#1}{x}{e}}}{x}{t}}
  The proof follows by induction on $s$ and some involved calculations. A mechanised
  version can be found in \cite[\coqe{Theorem Taylor}]{alvarez2020change}.
  \begin{itemize}
    \item When $s = x$ we have $\diffS{x}{x}{e} = e$ and so:
      \begin{align*}
        \lhs{x} = t + \epsilon e = \subst{x}{x}{t} + \epsilon \diff{x}
      \end{align*}
    \item When $s = y \neq x$ we have $\diffS{y}{x}{e} = 0$ and so:
      \begin{align*}
        \lhs{y} = y \diffEq y + 0 = \subst{y}{x}{y} + \epsilon \diff{y}
      \end{align*}
    \item When $s = 0$ we have $\LHS \diffEq 0 \diffEq \RHS$.
    \item When $s = s' + s''$ we have:
    \begin{align*}
      &\hspace{-20pt}\lhs{(s' + s'')} = \subst{s'}{x}{t + \epsilon e} + \subst{s''}{x}{t + \epsilon e}\\
      &\diffEq \pa{\subst{s'}{x}{t} + \epsilon\pa{\diff{s'}}}
        + \pa{\subst{s''}{x}{t} + \epsilon\pa{\diff{s''}}}\\
      &\diffEq \pa{\subst{s'}{x}{t} + \subst{s''}{x}{t}}
        + \pa{\diff{s'} + \diff{s''}}\\
      &= \subst{(s' + s'')}{x}{t} + \diff{(s' + s'')}
    \end{align*}
    \item When $s = \epsilon s'$ we have:
    \begin{align*}
      &\hspace{-30pt}\lhs{(\epsilon s')} = \epsilon \lhs{s'}\\
      &\diffEq \epsilon\pa{ \subst{s'}{x}{t} + \epsilon \pa{\diff{s'}} } \\
      &\diffEq \epsilon\pa{\subst{s'}{x}{t}} + \epsilon\pa{\epsilon \pa{\diff{s'}} } \\
      &\diffEq \subst{(\epsilon s')}{x}{t} + \epsilon \pa{\diff{(\epsilon s')}}
    \end{align*}
    \item When $s = s'\ s''$ we have:
    \begin{align*}
      &\hspace{-10pt}\lhs{(s'\ s'')} = (\lhs{s'})\ (\lhs{s''})\\
      &\diffEq \brak{\subst{s'}{x}{t} + \epsilon\pa{\diff{s'}}}\ 
        \brak{\subst{s''}{x}{t} + \epsilon\pa{\diff{s''}}}\displaybreak[0]\\
      &\diffEq \brak{\subst{s'}{x}{t}\ \pa{\subst{s''}{x}{t} + \epsilon\pa{\diff{s''}}}}\\
      \continued + \epsilon\brak{\pa{\diff{s'}}\ \pa{\subst{s''}{x}{t} + \epsilon\pa{\diff{s''}}} }\displaybreak[0]\\
      &\diffEq \pa{\subst{s'}{x}{t}\ (\subst{s''}{x}{t})} \\
      \continued + \epsilon\brak{\pa{\D(\subst{s'}{x}{t}) \cdot \pa{\diff{s''}}}
      \ (\subst{s''}{x}{t})}\\
      \continued + \epsilon\brak{\pa{\diff{s'}}\ \pa{\subst{s''}{x}{t + \epsilon e}}}\displaybreak[0]\\
      &= \subst{(s'\ s'')}{x}{t} \\
      \continued + \subst{\epsilon\brak{\pa{\D(s') \cdot
      \pa{\diffS{s''}{x}{e}}}\ s''}}{x}{t}\\
      \continued + \subst{\epsilon\brak{\pa{\diffS{s'}{x}{e}}\ \pa{\subst{s''}{x}{(x + \epsilon e)}}}}{x}{t}\displaybreak[0]\\
      &\diffEq \subst{(s'\ s'')}{x}{t} + \epsilon \pa{\diff{(s' + s'')}}
    \end{align*}
    \item When $s = \D(s') \cdot s''$ we have:
      \begin{align*}
        &\hspace{-10pt}{\lhs{(\D(s') \cdot s'')} = \D(\lhs{s'}) \cdot (\lhs{s''})}\\
        & \diffEq \D\pa{\subst{s'}{x}{t} + \epsilon\pa{\diff{s'}}} \cdot (\lhs{s''})\displaybreak[0]\\
        & \diffEq \dd{(\subst{s'}{x}{t})} \cdot (\lhs{s''}) + \epsilon
        \pa{\D\pa{\diff{s'}}\cdot (\lhs{s''})}\displaybreak[0]\\
        & \diffEq \dd{(\subst{s'}{x}{t})} \cdot \pa{\subst{s''}{x}{t} + \epsilon\diff{s''}}\\
        \continued + \epsilon\pa{\D\pa{\diff{s'}}\cdot (\lhs{s''})}\displaybreak[0]\\
        & \diffEq \dd{(\subst{s'}{x}{t})} \cdot \pa{\subst{s''}{x}{t}}
        + \epsilon\pa{\dd{(\subst{s'}{x}{t})} \cdot \pa{\diff{s''}}}\\
        \continued + \epsilon^2 \pa{\D(\dd{(\subst{s'}{x}{t})} \cdot \pa{\subst{s''}{x}{t}}) \cdot \pa{\diff{s''}}} \\
        \continued + \epsilon\pa{\D\pa{\diff{s'}}\cdot (\lhs{s''})}\displaybreak[0]\\
        & \diffEq \subst{\pa{\D(s') \cdot s''}}{x}{t}
        + \subst{\epsilon\pa{\dd{(s')} \cdot \pa{\diffS{s''}{x}{e}}}}{x}{t}\\
        \continued + \epsilon^2 \subst{\pa{\D(\dd{(s')} \cdot s'') \cdot \pa{\diffS{s''}{x}{e}}}}{x}{t} \\
        \continued + \epsilon\subst{\pa{\D\pa{\diffS{s'}{x}{e}}\cdot (\subst{s''}{x}{x + \epsilon e})}}{x}{t}\displaybreak[0]\\
        &\diffEq \subst{\pa{\D(s') \cdot s''}}{x}{t} + \epsilon
        \pa{\diff{(\dd{(s')} \cdot s'')}}
        \qedhere
      \end{align*}
  \end{itemize}
\end{proof}

The following lemmas relate differential substitution and standard
substitution, and will be of much use later.

\begin{lem}
  \label{lem:standard-differential-substitution}
  Whenever $x, y$ are (distinct) variables then for any unrestricted terms $t,
  u, v$ where $x$ is not free in $v$ we have:
  \[
    \subst{\pa{\diffS{t}{x}{u}}}{y}{v}
    = \diffS{\subst{t}{y}{v}}{x}{\subst{u}{y}{v}}
  \]
\end{lem}
\begin{proof}
  See \cite[\coqe{Lemma replace_dsubst}]{alvarez2020change}.
\end{proof}

\begin{lem}
  \label{lem:differential-standard-substitution}
  Whenever $x, y$ are (distinct) variables, with $y$ not free in either $u, v$,
  we have:
  \[
    \diffS{\subst{t}{y}{v}}{x}{u} \diffEq 
    \subst{\pa{\diffS{t}{x}{u}}}{y}{(\subst{v}{x}{x + \epsilon u})}
    + \subst{\pa{\diffS{t}{y}{\diffS{v}{x}{u}}}}{y}{v}
  \]
\end{lem}

One consequence of this ``syntactic Taylor's formula'' is that derivatives in
the difference $\lambda$-calculus can be computed by a sort of quasi-AD
algorithm: given an expression of the form $\lambda x . t$, its derivative at
point $s$ along $u$ can be computed by reducing the differential application
$\pa{\D(\lambda x . t) \cdot (u)}\ s$ which, as we shall see later, reduces (by
definition) to $\subst{\pa{\diffS{t}{x}{u}}}{x}{s}$. Alternatively, we can
simply evaluate $(\lambda x . t)\ (s + \epsilon(u))$ to compute $\subst{t}{x}{s
+ \epsilon(u)}$ which, by Theorem~\ref{thm:taylor} and
Lemma~\ref{lem:standard-differential-substitution}, is equivalent to
$\subst{t}{x}{s} + \epsilon(\subst{\diffS{t}{x}{u}}{x}{s})$. In an appropriate
setting (i.e. one where subtraction of terms is allowed and $\epsilon$ admits an
inverse) the derivative can then be extracted from this result by extracting the
term under the $\epsilon$. This process is remarkably similar to forward-mode
automatic differentiation, where derivatives are computed by adding
``perturbations'' to the program input.

Theorem~\ref{thm:taylor} also allows us to unpack all of the substitutions in
the definition of differential substitution. For example, the term
$\diffS{(t\ e)}{x}{u}$ can be expanded to:
\begin{align*}
  \diffS{(t\ e)}{x}{u}&= \brak{\D(t)\cdot \pa{\diffS{e}{x}{u}}\ e} 
  + \brak{\diffS{t}{x}{u}\ \pa{\subst{e}{x}{(x + \epsilon u)}}}
  \displaybreak[0]\\
  &\diffEq \brak{\D(t) \cdot \pa{\diffS{e}{x}{u}}\ e}
  + \brak{
    \diffS{t}{x}{u}\ \pa{
      e + \epsilon \diffS{e}{x}{u}
    }
  }
  \displaybreak[0]\\
  &\diffEq \brak{\D(t) \cdot \pa{\diffS{e}{x}{u}}\ e}
  + \brak{\diffS{t}{x}{u}\ e}
  + \epsilon \brak{\D\pa{\diffS{t}{x}{u}}\cdot \diffS{e}{x}{u}\ e} 
\end{align*}
We can generalise this procedure to arbitrary sequences of differential
substitutions, although the terms involved are too complex to give a simple
account.

\begin{lem}
  \label{lem:differential-substitution-dapp}
  For any basic terms $s^\b, t^\b$, variables $x_1, \ldots, x_n$ and basic terms
  $u^\b_1, \ldots, u^\b_n$ such that none of the $x_i$ appear free in the $u_i$,
  the differential substitution 
  \[ \frac{\partial^k
     (\D(s^\b)\cdot t^\b)}{\partial (x_1, \ldots, x_n)}
     (u^\b_1, \ldots, u^\b_n)
  \]
  is differentially equivalent to a sum of terms of the form
  \[ 
    \epsilon^z \D^l (v) \cdot (w_1, \ldots, w_l)
  \]
  where $v$ is of the form \[
    \frac{\partial^p t^\b}{\partial (x^{(t)}_1, \ldots,
    x^{(t)}_p)}(u^{(t)}_1, \ldots, u^{(t)}_p) 
  \] and every $w_j$ is of the form
    \[\frac{\partial^{q_j} e^\b}{\partial (x^{(w_j)}_1, \ldots,
    x^{(w_j)}_{q_j})}(u^{(w_j)}_1, \ldots, u^{(w_j)}_{q_j}) \] where $1 \leq l
    \leq 2^n$ and each pair of sequences $x^{(t)}_i, u^{(t)}_i$ corresponds to a
    reordering of some subsequence of the $x_i, u^\b_i$.
\end{lem}
\begin{proof}
  Straightforward induction on $k$ and applying Theorem~\ref{thm:taylor}.
\end{proof}

\begin{lem}
  \label{lem:differential-substitution-app}
  For any basic term $s^\b$ and additive term $t^*$, variables $x_1, \ldots, x_n$
  and basic terms $u^\b_1, \ldots, u^\b_n$ such that none of the $x_i$ appear
  free in the $u^\b_i$, the differential substitution
  \[ \frac{\partial^k
     (s^\b\ t^*)}{\partial (x_1, \ldots, x_n)}
     (u^\b_1, \ldots, u^\b_n)
  \]
  is differentially equivalent to a sum of terms of the form
  \[ 
    \epsilon^z \D^l (v) \cdot (w_1, \ldots, w_l)
  \]
  where $v$ is of the form \[
    \frac{\partial^p t^\b}{\partial (x^{(t)}_1, \ldots,
    x^{(t)}_p)}(u^{(t)}_1, \ldots, u^{(t)}_p) 
  \] and every $w_j$ is of the form
    \[\frac{\partial^{q_j} e^\b}{\partial (x^{(w_j)}_1, \ldots,
    x^{(w_j)}_{q_j})}(u^{(w_j)}_1, \ldots, u^{(w_j)}_{q_j}) \] where $1 \leq l
    \leq 2^n$ and each pair of sequences $x^{(t)}_i, u^{(t)}_i$ corresponds to a
    reordering of some subsequence of the $x_i, u^\b_i$.
\end{lem}

The above results may seem overly weak and arcane, but at its core they make a
very simple statement: if one applies any number of differential substitutions
to the term $\D(s) \cdot t$ (or $s\ t$) and ``cranks the lever'', pushing the
substitutions as far down the term as possible, then all the differential
substitutions in the resulting term are applied to either $s$ or $t$, and their
arguments are a reordering of some subsequence of the arguments to the initial
differential substitution.

\newcommand{\ddd}[3]{\frac{\partial^2 #1}{\partial {#2}^2}(#3)}
\begin{thm}
  \label{thm:differential-substitution-regular}
  Differential substitution is regular, that is, for any unrestricted terms $s,
  u, v$ where $x$ does not appear free in either $u$ or $v$, we have:
  \begin{align*}
    \diffS{s}{x}{0} &\diffEq 0\\
    \diffS{s}{x}{u + v} &\diffEq \diffS{s}{x}{u} 
    + \subst{\pa{\diffS{s}{x}{v}}}{x}{x + \epsilon u}
  \end{align*}
\end{thm}
\begin{proof}
  Both properties follow by induction on $s$. As the proof involves immense
  amounts of tedious calculations, we refer the reader to
  \cite[\coqe{Theorem Regularity}]{alvarez2020change}.
\end{proof}

\subsection{The Operational Semantics of \texorpdfstring{$\lambda_\epsilon$}{Lambda-Epsilon}}

With the substitution operations we have introduced so far, we can now
proceed to give a small-step operational semantics as a reduction system.

\begin{defi}
  The \textbf{one-step reduction relation}  $\reducesTo{} \subseteq \Lambda_\epsilon
  \times \Lambda_\epsilon$ is the least contextual relation satisfying the
  reduction rules in Figure~\ref{fig:reduction} below.
  \begin{figure}[ht]
    \[\def\arraystretch{1.3}
    \begin{array}{rcl}
      (\lambda x . t)\ s &\reducesTo_\beta& \subst{t}{x}{s}\\
      \D(\lambda x . t) \cdot s &\reducesTo_\partial& \lambda x . \pa{\diffS{t}{x}{s}}\\
    \end{array}
    \]
    \caption{One-step reduction rules for $\lambda_\epsilon$}
    \label{fig:reduction}
  \end{figure}

  We write $\reducesTo^+$ to denote the transitive closure of $\reducesTo$,
  and $\reducesTo^*$ to denote its transitive, reflexive closure.
\end{defi}

While the one-step reduction rules for $\lambda_\epsilon$ may seem identical to
those in the differential $\lambda$-calculus (see \cite{ehrhard2003differential}), 
they are in fact not equivalent, as our notion of differential substitution diverges substantially.

The above one-step reduction is defined as a relation from unrestricted terms to
unrestricted terms, but it is not compatible with differential equivalence. That
is to say, there may be differentially equivalent terms $t \diffEq t'$ such that
$t'$ can be reduced but $t$ cannot. For example, consider the term $(\lambda x
. x + 0)\ 0$, which contains no $\beta$-redexes that can be reduced. This term
is, however, equivalent to $(\lambda x . x)\ 0$, which clearly reduces to $0$. 

We could lift the one-step reduction relation to well-formed terms by setting
$\ul{t} \reducesTo \ul{t'}$ whenever there exist $s, s'$ such that $t \diffEq
s, t' \diffEq s'$ and $s \reducesTo s'$. This is not very satisfactory,
however, as it would make one-step reduction undecidable. Indeed, in order to
check whether $\ul{s} \reducesTo \ul{s'}$ it would be necessary to check
whether $s \reducesTo s'$ for all their (infinitely many) representatives $s,
s'$!

Another problem with this definition lies in the fact that the term $\ul{0}$
(ostensibly a value which should not reduce) can also be written as $\ul{0\ t}$
for any term $t$. Whenever $t$ reduces to $t'$ in one step, then according to
the previous definition so does $\ul{0\ t}$ reduce to $\ul{0\ t'}$, which is
equivalent to $\ul{0}$. Hence zero reduces to itself, rather than being a normal
form!

Fortunately the canonical form of a term $t$ gives us a representative of
$\ul{t}$ which is ``maximally reducible'', that is to say, whenever any
representative of $\ul{t}$ can be reduced, then so can $\can{t}$, possibly in
zero steps.

\begin{thm}
  \label{thm:reduction-canonical}
  Reduction is compatible with canonicalization. That is to say, if $s
  \reducesTo s'$, then $\can{s} \reducesTo^* s''$ for some $s'' \diffEq s'$.
\end{thm}
\begin{proof}
  We prove the theorem by induction on the depth at which reduction happens and
  the number of non-canonical $\beta$-redexes in the term (that is, the number
  of redexes of the form $(\lambda x . s)\ (t + \epsilon u)$). The
  most important cases are the ones where it happens at the outermost level, but
  we explicitly show some of the other cases. Before we do so, however, we state
  the following auxiliary properties:

  \begin{lem}\label{lem:canonical-lambda}
    Whenever $T \diffEq \lambda x . t$ where $T$ is a canonical term and $t$ is
    an unrestricted term, then $T$ is of the form $\sum_{i=1}^n \epsilon^{k_i}
    (\lambda x. t^\b_i)$, and additionally $t \diffEq \sum_{i=1}^n \epsilon^{k_i}
    t^\b_i$.
  \end{lem}
  \begin{lem}\label{lem:reduction-canonical-additive}
    Whenever $s + t \reducesTo^* e$ then $e = s' + t'$ with $s \reducesTo^* s'$
    and $t \reducesTo^* t'$. Whenever $\varepsilon s \reducesTo^* e$ then $e =
    \epsilon s'$ with $s \reducesTo^* s'$. In particular, whenever $\can{t}
    \reducesTo^* t'$ then $t' = \sum_{i=1}^n \epsilon^{k_i} t'_i$, where
    $\can{t} = \sum_{i=1}^n \epsilon^{k_i} t^\b_i$ and $t^\b_i \reducesTo^*
    t'_i$. Note that the $t'_i$ may not be basic terms and thus $t'$ may not be
    canonical.
  \end{lem}
  \begin{lem}\label{lem:reduction-cansum}
    Whenever $s \reducesTo^* s'$ and $t \reducesTo^* t'$, then $s \cansum t
    \reducesTo^* s' \cansum t'$.
  \end{lem}
  \begin{lem}\label{lem:reduction-reg}
    Whenever $s \reducesTo^* s'$ and $\can{t} \reducesTo^* t'$, then $\reg{s,
    \can{t}} \reducesTo^* \reg{s', t'}$.
  \end{lem}
  
  We proceed now to prove one of the cases where reduction happens in a subterm
  of $s$, which illustrates the ideas for the other cases:
  \begin{itemize}
    \item Let $s = \D(t) \cdot u$ and $s' = \D(t') \cdot u$, with $t
    \reducesTo t'$. By Lemma~\ref{lem:can-contextual}, we can write $\can{s}$
    as $\can{\D(\can{t}) \cdot \can{u}}$. Let $\can{t} = \sum_{i=1}^{n}
    \epsilon^{k_i} t^\b_i$. By induction and
    Lemma~\ref{lem:reduction-canonical-additive}, we have $t^\b_i \reducesTo^*
    t''_i$ and $t' \diffEq \sum_{i=1}^{n} \epsilon^{k_i} t''_i$. Applying the
    previous auxiliary lemmas, we obtain:
    \begin{align*}
      \can{s} 
      &= \can{\D(\can{t}) \cdot u}\\
      &= \cansum_{i = 1}^n (\epsilon^*)^{k_i}\reg{t^\b_i, \can{u}}\\
      &\reducesTo^* \cansum_{i = 1}^n (\epsilon^*)^{k_i}\reg{t''_i, \can{u}}\\
      &\diffEq \D\pa{\sum_{i=1}^n \epsilon^{k_i} t''_i} \cdot u\\
      &\diffEq \D(t') \cdot u
    \end{align*}
    \item Let $s = (t\ e)$ and $s' = (t\ e')$, with $e \reducesTo e'$. The
    result follows from the previous auxiliary lemmas and the fact that the
    primal $\textbf{pri}$ and tangent $\textbf{tan}$ components commute with
    reduction.
    \item Every other case is either immediate or follows from similar arguments.
  \end{itemize}
  
  The more involved cases are those when reduction happens at the outermost level of
  $s$. For brevity we will focus on the non-trivial cases where the underlying
  $\lambda$-abstraction involves only basic terms, as the more general cases
  follow by unfolding the $\lambda$-abstraction into a canonical sum and
  applying the primitive cases below to each summand separately.
  \begin{itemize}
    \item Let $s = \D(\lambda x . s^\b) \cdot t$ and $s' = \lambda x .
    \diffS{s^\b}{x}{t}$, and write $T$ for $\can{t}$. The proof proceeds then by
    induction on the number of summands of $T$.
    \begin{itemize}
      \item When $T = 0$ we have $\can{s} = 0$. On the other hand:
      \[s' = \lambda x . \diffS{s^\b}{x}{t} \diffEq \lambda x .
      \diffS{s^\b}{x}{0} \diffEq 0\]
      \item When $T = \epsilon^k t^\b + T'$, we apply the induction hypothesis
      and Lemma~\ref{lem:canonical-lambda} to obtain 
      \[ 
        \can{\D(\lambda x . s^\b) \cdot T'} 
        \reducesTo^*
        \sum_{i=1}^n\epsilon^{k_i}(\lambda x . w_i)
        \diffEq \lambda x . \diffS{s^\b}{x}{T'}
      \]
      Then the canonical form $\can{s}$ reduces as follows:
      \begin{align*}
        &\can{s}\\
        &= 
        (\epsilon^*)^k\D(\lambda x . s^\b) \cdot t^\b\\
        &\quad + \brak{\reg{\D(\lambda x . s^\b) \cdot T'}
        \cansum {\epsilon^* \D^*\pa{\reg{\D(\lambda x . s^\b) \cdot T'}} \cdot t^\b}}
        \displaybreak[0]\\
        &=
        (\epsilon^*)^k\D(\lambda x . s^\b) \cdot t^\b \\
        &\quad + \brak{\can{\D(\lambda x . s^\b) \cdot T'}
        \cansum {\epsilon^* \D^*\pa{\can{\D(\lambda x . s^\b) \cdot T'}} \cdot t^\b}}
        \displaybreak[0]\\
        &\reducesTo^*
        (\epsilon^*)^k \pa{\lambda x . \diffS{s^\b}{x}{t^\b}}\\
        &\quad + \brak{
        \pa{\sum_{i=1}^n \epsilon^{k_i} (\lambda x . w_i)}
        \cansum \epsilon^* \D^* \pa{
          \sum_{i=1}^n \epsilon^{k_i} (\lambda x . w_i)
        }\cdot t^\b}
        \displaybreak[0]\\
        &=
        (\epsilon^*)^k \pa{\lambda x . \diffS{s^\b}{x}{t^\b}}\\
        &\quad + \brak{
        \pa{\sum_{i=1}^n \epsilon^{k_i} (\lambda x . w_i)}
        \cansum \epsilon^* \pa{
          \sum_{i=1}^n \epsilon^{k_i} \D(\lambda x . w_i) \cdot t^\b
        }}
        \displaybreak[0]\\
        &\reducesTo^*
        (\epsilon^*)^k \pa{\lambda x . \diffS{s^\b}{x}{t^\b}}\\
        &\quad + \brak{
        \pa{\sum_{i=1}^n \epsilon^{k_i} (\lambda x . w_i)}
        \cansum \epsilon^* \pa{
          \sum_{i=1}^n \epsilon^{k_i} \lambda x . \diffS{w_i}{x}{t^\b}
        }}
        \displaybreak[0]\\
        &\diffEq \lambda x . \pa{\diffS{s^\b}{x}{\epsilon^k t^\b}}
        + 
        \lambda x . \pa{\diffS{s^\b}{x}{T'}}
        + \epsilon \pa{
          \lambda x . \diffS{\pa{\sum_{i=1}^n \epsilon^{k_i} w_i}}{x}{t^\b}
        }
        \displaybreak[0]\\
        &\diffEq \lambda x . \pa{\diffS{s^\b}{x}{\epsilon^k t^\b}}
        + 
        \lambda x . \pa{\diffS{s^\b}{x}{T'}}
        + \epsilon \pa{
          \lambda x . \diffS{\diffS{s^\b}{x}{T'}}{x}{t^\b}
        }
        \displaybreak[0]\\
        &\diffEq \lambda x . \pa{
          \diffS{s^\b}{x}{\epsilon^k t^\b}
          + \subst{\pa{\diffS{s^\b}{x}{T'}}}{x}{x + \epsilon t^\b}
        }
      \end{align*}
      By Theorem~\ref{thm:differential-substitution-regular}, this last term is
      precisely the term $\lambda x . \diffS{s^\b}{x}{\epsilon^k t^\b + T'}$,
      which is equivalent to $s'$ and thus the proof is concluded.
    \end{itemize}
    \item Let $s = (\lambda x . s^\b)\ t$ and $s' = \subst{s^\b}{x}{t}$, and
    suppose $\can{t} = t^* + \epsilon^* T$ (that is, $\textbf{pri}(\can{t}) = t^*$
    and $\textbf{tan}(\can{t}) = T$, with $t^*$ additive and $T$ canonical).
    \begin{align*}
      \can{s}
      &= (\lambda x . s^\b)\ t^* \cansum 
      \epsilon^*\textbf{ap}\pa{\reg{\lambda x . s^\b, T}, t^*}
    \end{align*}
    Since the term $(\lambda x . s^\b) \cdot T$ contains one less non-canonical
    $\beta$-redex than the term $s$, we apply our induction hypothesis to obtain
    \[
      \reg{\lambda x . s^\b, T} =
      \can{(\lambda x . s^\b) \cdot T }
      \reducesTo^* \sum_{i=1}^n \epsilon^{k_i} (\lambda x . w_i)
      \diffEq \lambda x . \diffS{s^\b}{x}{T}
    \]
    and therefore $\diffS{s^\b}{x}{T} \diffEq \sum_{i=1}^n \epsilon^{k_i} w_i$. 
    With this in mind we continue to reduce the previous equation:
    \begin{align*}
      &(\lambda x . s^\b)\ t^* \cansum 
      \epsilon^*\textbf{ap}\pa{\reg{\lambda x . s^\b, T}, t^*}
      \displaybreak[0]\\
      &\quad \reducesTo^*
      \subst{s^\b}{x}{t^*} 
      \cansum
      \epsilon^*\textbf{ap}\pa{\sum_{i=1}^n \epsilon^{k_i} (\lambda x . w_i), t^*}
      \displaybreak[0]\\
      &\quad = \subst{s^\b}{x}{t^*}
      \cansum \epsilon^*\pa{\sum_{i=1}^n \epsilon^{k_i} (\lambda x . w_i)\ t^*}
      \displaybreak[0]\\
      &\quad \reducesTo^* 
      \subst{s^\b}{x}{t^*}
      \cansum \epsilon^*\pa{\sum_{i=1}^n \epsilon^{k_i} (\subst{w_i}{x}{t^*})}
      \displaybreak[0]\\
      &\quad = 
      \subst{s^\b}{x}{t^*}
      \cansum \epsilon^*\brak{\subst{\pa{\sum_{i=1}^n \epsilon^{k_i} w_i}}{x}{t^*})}
      \displaybreak[0]\\
      &\quad \diffEq
      \subst{s^\b}{x}{t^*}
      + \epsilon\pa{\subst{\diffS{s^\b}{x}{T}}{x}{t^*}}
      \displaybreak[0]\\
      &\quad \diffEq
      \subst{s^\b}{x}{t^* + \epsilon T}
      \displaybreak[0]\\
      &\quad \diffEq \subst{s^\b}{x}{t}\qedhere
    \end{align*}
  \end{itemize}
\end{proof}

The above result then legitimises our proposed ``existential'' definition of
reduction of well-formed terms, as it shows that, in order to reduce a given
term, it suffices to reduce its canonical form. It also gets rid of the
``reducing zero'' problem, as canonical forms do not contain ``spurious''
representations of zero.

\begin{defi}
  Given well-formed terms $\ul{s}, \ul{s'}$, we say that $\ul{s}$
  \textbf{reduces to} $\ul{s'}$ \textbf{in one step}, and write $\ul{s}
  \mathrel{\ul{\reducesTo}} \ul{t}$, whenever $\can{s}
  \reducesTo s''$ and $s'' \diffEq s'$, for some canonical form $\can{s}$ of
  $\ul{s}$.
\end{defi}

\begin{prop}
  \label{prop:red-wf-can}
  Whenever $\ul{s} \wfReducesTo \ul{s'}$ then for any term $\ul{t}$ we have
  $\ul{s + t} \wfReducesTo \ul{s' + t}$. 
  
  If $t=t^*$ is an additive term, then additionally $\ul{s\ t^*} \wfReducesTo^+
  \ul{s'\ t^*}$.
  
  Furthermore, when $t=t^\b$ is a basic term (in particular $t^\b$ is not
  differentially equivalent to zero), we also have $\ul{\D(s) \cdot t}
  \wfReducesTo^+ \ul{\D(s') \cdot t}$.

  Conversely, whenever $s$ is not differentially equivalent to zero and $\ul{t}
  \wfReducesTo \ul{t'}$, then $\ul{s\ t} \wfReducesTo^+ \ul{s\ t'}$ and $\ul{\D(s)
  \cdot t} \wfReducesTo^+ \ul{\D(s) \cdot t'}$.
\end{prop}

The wording of the above definition specifies that a well-formed
term reduces to another whenever \emph{any} of its canonical forms reduces.
As we have shown before, canonical forms are in fact only unique up to
commutativity of addition and derivatives. Addition is not problematic, since
it respects reduction; that is to say, if a sum $s + t$ reduces to $s' + t'$
then its permutation $t + s$ also reduces to $t' + s'$. Symmetry of
derivatives raises a more significant issue: consider the following diagram,
which does \emph{not} commute:
\begin{center}
\begin{tikzpicture}
  \node (a) at (0, 0) {$\D (\D (\lambda x . t) \cdot u) \cdot v$};
  \node (b) at (5, 0) {$\D (\D (\lambda x . t) \cdot v) \cdot u$};
  \node (a') at (0, -2) {$\D \pa{\lambda x . \diffS{t}{x}{u}} \cdot v$};
  \node (b') at (5, -2) {$\D \pa{\lambda x . \diffS{t}{x}{v}} \cdot u$};
  \draw[white] (a) edge node[black]{$\diffEq$} (b);
  \draw[white] (a) edge node[rotate=-90,black]{$\reducesTo$} (a');
  \draw[white] (b) edge node[rotate=-90,black]{$\reducesTo$} (b');
  \draw[white] (a') edge node[black]{$\not\diffEq$} (b');
  \end{tikzpicture}
\end{center}
One-step reduction is still computable, since the set of canonical forms of
any given term is finite, but we will have to keep this behaviour in mind when
showing confluence.

A proof of confluence for $\lambda_\epsilon$ will proceed by the standard
Tait/Martin-L\"of method by introducing a notion of parallel reduction on terms.

\begin{defi}
  The \textbf{parallel reduction} relation between (unrestricted) terms is
  defined according to the deduction rules in Figure~\ref{fig:par-reduction}.
  \begin{figure}[ht]
    \begin{center}
    \bpt{
      \AxiomC{}
      \LeftLabel{($\parTo_x$)}
      \UnaryInfC{$x \parTo x$}
    }
    \bpt{
      \AxiomC{}
      \LeftLabel{($\parTo_0$)}
      \UnaryInfC{$0 \parTo 0$}
    }
    \bpt{
      \AxiomC{$t \parTo t'$}
      \LeftLabel{($\parTo_\lambda$)}
      \UnaryInfC{$\lambda x . t \parTo \lambda x . t'$}
    }
    \end{center}
    \begin{center}
    \bpt{
      \AxiomC{$t \parTo t'$}
      \LeftLabel{($\parTo_\varepsilon$)}
      \UnaryInfC{$\epsilon t \parTo \epsilon t'$}
    }
    \bpt{
      \AxiomC{$s \parTo s'$}
      \AxiomC{$t \parTo t'$}
      \LeftLabel{($\parTo_+$)}
      \BinaryInfC{$s + t \parTo s' + t'$}
    }
    \end{center}\begin{center}
    \bpt{
      \AxiomC{$s \parTo s'$}
      \AxiomC{$t \parTo t'$}
      \LeftLabel{($\parTo_{\textbf{ap}}$)}
      \BinaryInfC{$s\ t \parTo s'\ t'$}
    }
    \bpt{
      \AxiomC{$s \parTo s'$}
      \AxiomC{$t \parTo t'$}
      \LeftLabel{($\parTo_\D$)}
      \BinaryInfC{$\D(s)\cdot t \parTo \D(s') \cdot t'$}
    }
    \end{center}
    \begin{center}
    \bpt{
      \AxiomC{$s \parTo \lambda x . s'$}
      \AxiomC{$t \parTo t'$}
      \LeftLabel{($\parTo_\beta$)}
      \BinaryInfC{$s\ t \parTo \subst{s'}{x}{t'}$}
    }
    \bpt{
      \AxiomC{$s \parTo \lambda x . s'$}
      \AxiomC{$t \parTo t'$}
      \LeftLabel{($\parTo_\partial$)}
      \BinaryInfC{$\D(s)\cdot t \parTo \lambda x . \diffS{s'}{x}{t'}$}
    }
    \end{center}
    \caption{Parallel reduction rules for $\lambda_\epsilon$}
    \label{fig:par-reduction}
  \end{figure}

  The parallel reduction relation can be extended to well-formed terms by
  setting $\ul{t} \wfParTo \ul{t'}$ whenever $\can{t} \parTo t''$ with $t''
  \diffEq t'$ for some canonical form of $\ul{t}$.
\end{defi}

\begin{rem}
  Our definition of parallel reduction differs slightly from the usual in the
  rule $(\parTo_\beta)$, which allows reducing a newly-formed
  $\lambda$-abstraction. This is necessary because our calculus contains terms
  of the shape $(\D(\lambda x . s) \cdot u)\ t$, which we need to parallel reduce
  in a single step to $\subst{\pa{\diffS{s}{x}{u}}}{x}{t}$. The original
  presentation of the differential $\lambda$-calculus opted instead for 
  adding an extra parallel reduction rule to allow for the case of reducing an
  abstraction under a differential application. Similarly, our rule
  $(\parTo_\partial)$ allows reducing terms of the form $\D (\D(\lambda x . s) \cdot
  u) \cdot v$ in a single step.
\end{rem}

One convenient property of the parallel reduction relation lies in its relation
to canonical forms. As we saw in Theorem~\ref{thm:reduction-canonical},
canonical forms are ``maximally reducible'', but don't respect the number of
reduction steps. This is no longer the case for parallel reduction: the
process of canonicalization only duplicates regexes ``in parallel'' (that is, by
copying them onto multiple separate summands) or in a ``parallelizable series''
(i.e. a differential application may be regularized into a term of the form 
$\D (\D(\ldots) \cdot u) \cdot v$, which can be entirely reduced in a single
parallel reduction step). 

\begin{thm}
  \label{thm:parallel-reduction-canonical}
  Whenever $s \parTo s'$, then $\can{s} \parTo s''$ for some $s'' \diffEq s'$.
\end{thm}
\begin{proof}
    It suffices to inspect the proof of
    Theorem~\ref{thm:reduction-canonical} and convince oneself that all of the
    reductions introduced by the proof can be lifted into a single instance of
    parallel reduction.
\end{proof}

We also state the following standard properties of parallel reduction, all of which
can be proven by straightforward induction on the term.

\begin{lem}
  Parallel reduction sits between one-step and many-step reduction. That is to
  say: $\reducesTo {}\subseteq {}\parTo {}\subseteq {}\reducesTo^*$, and
  furthermore $\mathrel{\ul{\reducesTo}}{} \subseteq {} \wfParTo {} \subseteq {} \mathrel{
    \ul{\reducesTo}}^*$.
\end{lem}

\begin{lem}
  The parallel reduction relation is contextual. In particular, every term
  parallel-reduces to itself.
\end{lem}

\begin{lem}
  Parallel reduction cannot introduce free variables. That is to say: whenever
  $t \parTo t'$, we have $\text{FV}(t') \subseteq \text{FV}(t)$.
\end{lem}

\begin{lem}
  Whenever $\lambda x . t \parTo u$, it must be the case that $u = \lambda x .
  t'$ and $t \parTo t'$.
\end{lem}

\begin{lem}
  \label{lem:par-to-substitution}
  Whenever $s \parTo s'$ and $t \parTo t'$ then $\subst{s}{x}{t} \parTo
  \subst{s'}{x}{t'}$, and furthermore there is some $w$ with $\diffS{s}{x}{t}
  \parTo w \diffEq \diffS{s'}{x}{t'}$.
\end{lem}
\begin{proof}
  Both proofs follow by induction on the derivation applied to obtain $s \parTo
  s'$. We explicitly prove some non-trivial cases. First, for standard
  substitution.
  \begin{itemize}
    \item When the last rule applied is $(\parTo_\beta)$, that 
    is: $s = e\ u, s' = \subst{e'}{y}{u'}$, with $e \parTo \lambda y . e', u
    \parTo u'$. By the induction hypothesis, we have $\subst{e}{x}{t} \parTo
    \lambda y . (\subst{e'}{x}{t'})$ and $\subst{u}{x}{t} \parTo \subst{u'}{x}{t'}$, hence:
    \begin{align*}
      \subst{s}{x}{t} &= \subst{e}{x}{t}\ (\subst{u}{x}{t})
      \parTo \subst{(\subst{e'}{x}{t'})}{y}{(\subst{u'}{x}{t'})}
      = \subst{(\subst{e'}{y}{u'})}{x}{t'}
    \end{align*}
    \item When the last rule applied is $(\parTo_\partial)$, that is:
    $s = \D(e) \cdot u$, $s' = \lambda y . \diffS{e'}{x}{u'}$, with
    $e \parTo \lambda y . e', u \parTo u'$. As before, we apply the induction hypothesis and
    obtain:
    \begin{align*}
      \subst{s}{x}{t} &= D(\subst{e}{x}{t}) \cdot (\subst{u}{x}{t})
      \displaybreak[0]\\
      &\parTo \lambda y . \pa{\diffS{(\subst{e'}{x}{t'})}{y}{\subst{u'}{x}{t'}}}
      = \subst{\pa{\lambda y . \diffS{e'}{y}{u'}}}{x}{t'}
    \end{align*}
  \end{itemize}
  The corresponding cases for differential substitution are slightly more technically
  involved.
  \begin{itemize}
    \item When the last rule applied is $(\parTo_\beta)$, that is: $s = (e\ u),
    s' = \subst{e'}{y}{u'}$, with $e \parTo \lambda y . e', u \parTo u'$. By the
    induction hypothesis and applying the definition of differential
    substitution, we have $\diffS{e}{x}{t} \parTo \lambda y . \diffS{e'}{x}{t'}$
    and $\diffS{u}{x}{t} \parTo \diffS{u'}{x}{t'}$. By applying the previous
    proof we also obtain $\subst{u}{x}{x + \epsilon t} \parTo \subst{u'}{x}{x +
    \epsilon t'}$ hence\footnote{
    Observe that the reasoning here would not hold if we had opted to define
    parallel reduction in the ``standard'' way, as differential substitution may
    ``unfold'' an application into a differential application followed by a
    standard one.
    }:
    \begin{align*}
      \diffS{s}{x}{t} &= 
      \brak{\pa{\D(e) \cdot \pa{\diffS{u}{x}{t}}}\ u}
      + \brak{\diffS{e}{x}{t}\ (\subst{u}{x}{x + \epsilon t})}
      \displaybreak[0]\\
      &\parTo
      \subst{\brak{\diffS{e'}{y}{\diffS{u'}{x}{t'}}}}{y}{u'}
      + \subst{\pa{\diffS{e'}{x}{t'}}}{y}{(\subst{u'}{x}{x + \epsilon t'})}
    \end{align*}

    On the other hand, since $y$ is not free in either $u'$ or $t'$, applying
    Lemma~\ref{lem:differential-standard-substitution}, we obtain:
    \begin{align*}
      \diffS{s'}{x}{t'} &= \diffS{\pa{\subst{e'}{y}{u'}}}{x}{t'}\\
      &\diffEq 
        \subst{\pa{\diffS{e'}{x}{t'}}}{y}{(\subst{u'}{x}{x + \epsilon t'})} 
        + \subst{\pa{\diffS{e'}{y}{\diffS{u'}{x}{t'}}}}{y}{u'}
    \end{align*}
    \item When the last rule applied is $(\parTo_\partial)$, that is:
    $s = \D(\lambda y . e) \cdot u$, $s' = \lambda y . \diffS{e'}{x}{u'}$, with
    $e \parTo e', u \parTo u'$. As before, we apply the induction hypothesis and
    obtain:
    \begin{align*}
      \subst{s}{x}{t} &= D(\lambda y . \subst{e}{x}{t}) \cdot (\subst{u}{x}{t})\\
      &\parTo   \lambda y . \pa{\diffS{(\subst{e'}{x}{t'})}{y}{\subst{u'}{x}{t'}}}\\
      &= \subst{\pa{\lambda y . \diffS{e'}{y}{u'}}}{x}{t'}\qedhere
    \end{align*}
  \end{itemize}
\end{proof}

We first prove that parallel reduction has the diamond property when applied to
canonical terms, taking care that it holds up to differential equivalence (note
that, much like one-step reduction, the result of parallel-reducing a canonical
term need not be canonical). For this, we introduce the usual notion of a full
parallel reduct of a term.

\newcommand{\fpr}[0]{_\da}
\begin{defi}
  Given an unrestricted term $t$, its \textbf{full parallel reduct}
  $t\fpr$ is defined inductively by:
  \[\def\arraystretch{1.5}
  \begin{array}{rcl}
    x\fpr &\defeq& x \\
    (\epsilon t) \fpr &\defeq& \epsilon (t\fpr)\\
    (s + t)\fpr &\defeq& (s\fpr) + (t\fpr)\\
    0\fpr &\defeq& 0\\
    (\lambda x . t) \fpr &\defeq& \lambda x . (t\fpr)\\
    (s\ t)\fpr &\defeq& \begin{cases}
      \subst{e}{x}{t\fpr} &\text{ if $s\fpr = \lambda x . e$}\\
      (s\fpr)\ (t\fpr)&\text{ otherwise}
    \end{cases}\\
    (\D(s)\cdot t)\fpr &\defeq& \begin{cases}
      \lambda x . \diffS{e}{x}{t\fpr} &\text{ if $t\fpr = \lambda x . e$}\\
      D(s\fpr)\cdot (t\fpr) &\text{ otherwise}
    \end{cases}
  \end{array}
  \]
\end{defi}

\begin{lem}
  \label{lem:full-reduction-lambda}
  Whenever $s \parTo \lambda x . v$, then $s\fpr$ is of the form $\lambda x .
  w$, for some term $w$.
\end{lem}
\begin{proof}
  The proof follows by inspection of the parallel reduction rules. Consider a
  derivation of $s \parTo s'$, which will be of the form
  $\lambda x_1.\lambda x_2.\ldots \lambda x_n . t$ (with $n$ possibly equal to
  $0$). In general the amount of abstractions at the outermost level of the term
  depends on our choice of a derivation for $s \parTo s'$. Suppose then that we
  pick a derivation for which $n$ is maximal. If this derivation does not
  use the rules $\parTo_\textbf{ap}$ or $\parTo_\D$, then it is already
  the case that $s' = s\fpr$. On the other hand, if it contains either rule, it
  is straightforward to see that replacing the last application of
  $\parTo_\textbf{ap}$ or $\parTo_\D$ by $\parTo_\beta$ or $\parTo_\partial$
  respectively the resulting term has at least as many $\lambda$-abstractions at
  the outermost level as the previous one. Iterating this process, we obtain
  that the number of outermost abstractions is maximised precisely whenever $s'
  = s\fpr$.
\end{proof}

\begin{thm}
  \label{thm:full-reduction}
  For any unrestricted terms $s, s'$ such that $s \parTo s'$, there is an
  unrestricted term $w$ such that $s' \parTo w$ and $w \diffEq s\fpr$.
\end{thm}
\begin{proof}
  The proof follows by induction on the derivation of $s \parTo s'$. Most cases
  are straightforward, and the rest follow as a corollary of
  Lemma~\ref{lem:par-to-substitution}, as we now show.
  \begin{itemize}
    \item The last applied rule is $\parTo_\beta$, that is, $s = (t\ e)$, 
      $s' = \subst{t'}{x}{e'}$, with $t \parTo \lambda x . t', e \parTo e'$.
    
      By the induction hypothesis we have $e' \parTo w_e \diffEq e\fpr$ and
      $(\lambda x . t') \parTo (\lambda x . w_t) \diffEq t\fpr$. 
      By Lemma~\ref{lem:full-reduction-lambda}, it follows that $t\fpr$ has the
      form $\lambda x . v_t$, and since $(\lambda x . w_t) \diffEq (\lambda x .
      v_t)$ we also have $w_t \diffEq v_t$.
      Then: \[
        s' = \subst{t'}{x}{e'} \parTo \subst{w_t}{x}{w_e}
        \diffEq \subst{v_t}{x}{e\fpr} = (t\ e)\fpr
      \]

      \item The last applied rule is $\parTo_\partial$, that is,
      $s = (\D(t)\cdot u), s' = \lambda x . \diffS{t'}{x}{u'}$, with
      $t \parTo \lambda x . t', u \parTo u'$.

      By the induction hypothesis we have $u' \parTo w_u \diffEq u\fpr$ and $\lambda x . t'
      \parTo \lambda x . w_t \diffEq t\fpr$. Again we must have $t\fpr = \lambda
      x . v_t$ with $w_t \diffEq v_t$. Then: \[
        s' = \diffS{t'}{x}{u'} \parTo \diffS{w_t}{x}{w_u}
        \diffEq \diffS{v_t}{x}{u\fpr} = (\D(t) \cdot u)\fpr\qedhere
      \]
  \end{itemize}
\end{proof}

\begin{cor}
  Parallel reduction has the diamond property up to differential equivalence.
  That is to say, for any unrestricted term $t$ and terms $t_1, t_2$ such that $t
  \parTo t_1$ and $t \parTo t_2$, there are terms $u, v$ making the following
  diagram commute:

  \begin{center}
  \begin{tikzpicture}
    \node (a) at (0, 0) {$T$};
    \node (t1) at (-1, -1) {$t_1$};
    \node (t2) at (1, -1) {$t_2$};
    \node (u) at (-1, -2.5) {$u$};
    \node (v) at (1, -2.5) {$v$};
    \draw[white] (a) edge node[rotate=225,black]{$\parTo$} (t1);
    \draw[white] (a) edge node[rotate=-45,black]{$\parTo$} (t2);
    \draw[white] (t1) edge node[rotate=-90,black]{$\parTo$} (u);
    \draw[white] (t2) edge node[rotate=-90,black]{$\parTo$} (v);
    \draw[white] (u) edge node[black]{$\diffEq$} (v);
    \end{tikzpicture}
  \end{center}
\end{cor}

\begin{lem}
  \label{lem:perm-eq-fpr}
  Given unrestricted terms $s \permEq s'$ which are permutatively equivalent,
  that is, which differ only up to a reordering of their additions and
  differential applications, their full parallel reducts are differentially
  equivalent.
\end{lem}
\begin{proof}
  The proof follows by straightforward induction on the proof of permutative
  equivalence. The only involved case is $s = \D (\D(\lambda x . t) \cdot u)
  \cdot v \permEq \D(\D(\lambda x . t) \cdot v)\cdot u = s'$. But in this case
  the result follows as a corollary of
  Proposition~\ref{prop:differential-substitution-commute}
  \begin{align*}
    s\fpr &= (\D (\D(\lambda x . t)\cdot u)\cdot v)\fpr
    \displaybreak[0]\\
    &=       \lambda x . \dd{t}{x}{u, v}
    \displaybreak[0]\\
    &\diffEq \lambda x . \dd{t}{x}{v}{u}\\
    &= s'\fpr\qedhere
  \end{align*}
\end{proof}

\begin{thm}
  The reduction relation $\wfParTo$ has the diamond property. That is, whenever $\ul
  s \wfParTo \ul u$ and $\ul s \wfParTo \ul v$ there is a term
  $\ul c$ such that $\ul u \wfParTo \ul c$ and $\ul v \wfParTo \ul c$.
\end{thm}
\begin{proof}
  Consider a well-formed term $\ul{s}$, and suppose that $\ul{s} \wfParTo
  \ul{u}$ and $\ul{s} \wfParTo \ul{v}$. In particular, this means there are two
  canonical forms $\can{s}_1, \can{s}_2$ of $s$ such that $\can{s}_1 \parTo u$
  and $\can{s}_2 \parTo v$. These canonical forms $\can{s}_1, \can{s}_2$
  are equivalent up to permutative equivalence, and so their full parallel
  reducts are differentially equivalent as per Lemma~\ref{lem:perm-eq-fpr}.
  Denote their $\diffEq$-equivalence class by $\ul{c}$.
  Therefore since $\can{s}_1 \parTo \can{s}_{1\da} \diffEq c$ and $\can{s}_2 \parTo 
  \can{s}_{2\da} \diffEq c$ it follows that $\ul{u} \wfParTo \ul{c}$ and $\ul{v}
  \wfParTo \ul{c}$.
\end{proof}

\begin{cor}
  The reduction relation $\ul{\reducesTo}$ is confluent.
\end{cor}

\subsection{Encoding the Differential \texorpdfstring{$\lambda$}{Lambda}-Calculus}

It is immediately clear, from simply inspecting the operational semantics for
$\lambda_\epsilon$, that it is closely related to the differential
$\lambda$-calculus -- indeed, every Cartesian differential category is a
Cartesian difference category, and this connection should also be reflected in
the syntax.

As it turns out, there is a clean translation that embeds $\lambda_\epsilon$
into the differential $\lambda$-calculus, which proceeds by deleting every term
that contains an $\epsilon$. The intuition behind this scheme should be
apparent: every single differential substitution rule in $\lambda_\epsilon$ is
identical to the corresponding case for the differential $\lambda$-calculus,
once all the $\epsilon$ terms are cancelled out.

\begin{defi}
    \label{def:erasure}
  Given an unrestricted $\lambda_\epsilon$ term $t$, its
  \textbf{$\epsilon$-erasure} is the differential $\lambda$-term $\ceil{t}$ defined inductively
  according to the rules in Figure~\ref{fig:erasure} below.

  \begin{figure}[htbp]
  \[
  \begin{array}{rcl}
    \ceil{x}&\defeq&x\\
    \ceil{0}&\defeq&0\\
    \ceil{s+t}&\defeq&\ceil{s} + \ceil{t}\\
    \ceil{\epsilon t}&\defeq& 0\\
    \ceil{s\ t}&\defeq&\ceil{s}\ \ceil{t}\\
    \ceil{\D(s)\cdot t}&\defeq&\D(\ceil{s})\cdot \ceil{t}
  \end{array}{rcl}
  \]
  \caption{$\epsilon$-erasure of a term $t$}
  \label{fig:erasure}
\end{figure}
\end{defi}

\begin{prop}
  The erasure $\ceil{t}$ is invariant under differential equivalence. That is to
  say, whenever $t\diffEq t'$, it is the case that $\ceil{t} = \ceil{t'}$.
\end{prop}
\begin{proof}
  Follows immediately from inspecting the differential equivalence rules
  in Figure~\ref{fig:differential-equivalence} and noticing that the erasure of
  both sides coincides.
  
  Note that the standard presentation of the differential
  $\lambda$-calculus does not distinguish between equivalent terms, so terms
  like $s + t$ and $t + s$ are not merely equivalent but in fact identical.
\end{proof}

\begin{prop}
  Erasure is compatible with standard and differential substitution. That is to
  say, for any terms $s, t$ and a variable $x$, we have:
  \[
    \ceil{\subst{s}{x}{t}} = \subst{\ceil{s}}{x}{\ceil{t}}\\
    \ceil{\diffS{s}{x}{t}} = \diffS{\ceil{s}}{x}{\ceil{t}}\\
  \]
\end{prop}

\begin{cor}
  Whenever $s \reducesTo s'$, then $\ceil{s} \reducesTo^* \ceil{s'}$.
\end{cor}

These results form the syntactic obverse to the purely semantic fact that every Cartesian
differential category is (trivially) a Cartesian difference category (see 
\cite[Section 4.3]{alvarez2020extended} for a proof): the former exhibits the differential
$\lambda$-calculus as an instance of $\lambda_\epsilon$ where the $\epsilon$
operator is ``crossed out'', whereas the later shows that every Cartesian
differential category can be understood as a ``degenerate'' Cartesian difference
category, in the same sense that the corresponding infinitesimal extension is
just the zero map.

\section{Simple Types for \texorpdfstring{$\lambda_\epsilon$}{Lambda-Epsilon}}
\label{sec:simple-types}

Much like the differential $\lambda$-calculus, $\lambda_\epsilon$ can be endowed
with a system of simple types, built from a set of basic types using the usual
function type constructor. 

\begin{defi} The set of \textbf{types} and \textbf{contexts} of the
  $\lambda_\epsilon$-calculus is given by the following inductive definition:
  \[
    \begin{array}{rccl}
      \text{Types: } &\sigma, \tau& \defeq & \mathbf{t} \alter \sigma \Ra \tau\\
      \text{Contexts: } &\Gamma&\defeq& \emptyset \alter \Gamma, x : \tau
    \end{array}
  \]
  assuming a countably infinite set of basic types $\mathbf{t}, \mathbf{s}
  \ldots$ is given. 
\end{defi}

The typing rules for the $\lambda_\epsilon$-calculus are given in
Figure~\ref{fig:typing-rules} below, and should not be in the least surprising,
as they are identical to the typing rules for the differential
$\lambda$-calculus, with the addition of a typing rule for the infinitesimal
extension of a term. As one would expect, our type system  enjoys all the
``usual'' structural properties and their proofs follow by straightforward
induction on the typing derivation. Note, however, that all of these typing
rules operate on unrestricted terms, rather than on well-formed terms, for
reasons that we will clarify later.

\begin{figure}[ht]
  \begin{center}
  \bpt{
    \AxiomC{}
    \UnaryInfC{$\Gamma, x : \tau \ts x : \tau$}
  }
  \bpt{
    \AxiomC{$\Gamma \ts s : \tau \Ra \sigma$}
    \AxiomC{$\Gamma \ts t : \tau$}
    \BinaryInfC{$\Gamma \ts (s\ t) : \sigma$}
  }
  \bpt{
    \AxiomC{$\Gamma, x : \tau \ts t : \sigma$}
    \UnaryInfC{$\Gamma \ts \lambda x . t : \sigma \Ra \tau$}
  }
  \end{center}
  \begin{center}
  \bpt{
    \AxiomC{}
    \UnaryInfC{$\Gamma \ts 0 : \tau$}
  }
  \bpt{
    \AxiomC{$\Gamma \ts s : \tau$}
    \AxiomC{$\Gamma \ts t : \tau$}
    \BinaryInfC{$\Gamma \ts s + t : \tau$}
  }
  \bpt{
    \AxiomC{$\Gamma \ts t : \tau$}
    \UnaryInfC{$\Gamma \ts \epsilon t : \tau$}
  }
  \end{center}
  \begin{center}
  \bpt{
    \AxiomC{$\Gamma \ts s : \tau \Ra \sigma$}
    \AxiomC{$\Gamma \ts t : \tau$}
    \BinaryInfC{$\Gamma \ts \D (s) \cdot t : \tau \Ra \sigma$}
  }
  \end{center}
  \caption{Simple types for $\lambda_\epsilon$}
  \label{fig:typing-rules}
\end{figure}

According to the above rules, typing derivations are
\emph{invertible}, that is to say, whenever $\Gamma \ts t : \tau$ and $t$ is of
the form $s + e$, then it must be the case that $\Gamma \ts s : \tau$ and
$\Gamma \ts e : \tau$, and so on. One property that fails to hold is uniqueness
of typings: indeed the term $0$ admits any type, as do terms such as $0 + 0$ or
$(\lambda x . 0)\ y$.

The following ``standard'' properties also hold, and can be proven by
straightforward induction on the relevant typing derivation.

\begin{prop}[Weakening]
  Whenever $\Gamma \ts t : \tau$, then for any context $\Sigma$ which is
  disjoint with $\Gamma$ it is also the case that $\Gamma, \Sigma \ts t : \tau$.
\end{prop}

\begin{prop}[Substitution]
  Whenever $\Gamma, x : \tau \ts s : \sigma$ and $\Gamma \ts t : \tau$, we have:
  \begin{enumerate}[(i)]
    \item $\Gamma \ts \subst{s}{x}{t} : \sigma$
    \item $\Gamma, x : \tau \ts \diffS{s}{x}{t} : \sigma$
  \end{enumerate}
\end{prop}

\begin{thm}[Subject reduction]
  \label{thm:subject-reduction}
  Whenever $\Gamma \ts t : \tau$ and $t \reducesTo t'$ then $\Gamma \ts t' : \tau$.
\end{thm}

Since we have defined well-formed terms as equivalence classes of unrestricted
terms, we might ask if typing is compatible with this equivalence relation. The
answer is unfortunately no, that is to say, there are ill-typed terms that are
differentially equivalent to well-typed terms. In particular, the term $(0\ t)$
is differentially equivalent to the term $0$, but while the later is trivially
well-typed, the former will not be typable for many choices of $t$ (for example,
whenever $t = (x\ x)$). A weaker version of this property does hold, however,
that makes use of canonicity.

\begin{prop}
  \label{prop:typing-canonical}
  Whenever $\Gamma \ts t : \tau$, then $\Gamma \ts \can{t} : \tau$, and
  furthermore whenever $\Gamma \ts \can{t} : \tau$ then every canonical form of
  $t$ admits the same type. 
\end{prop}
\begin{proof}
  The proof proceeds by induction on the typing derivation, by noting that every
  operation involved in canonicalization respects the typing rules. 
\end{proof}

\begin{rem}
  The above issue could have been entirely avoided by circumventing the untyped
  calculus altogether and instead defining and operating on well-typed (unrestricted) 
  terms directly. We have preferred to work out the untyped case first for two reasons:
  first, to mimic the development of the differential $\lambda$-calculus. Second, since
  differentiation of control and fixpoint operators is suspect (in that there is not
  an ``obvious'' choice of a derivative for them), we hope that working in an untyped
  calculus featuring Church encodings and a Y combinator can illustrate what the
  ``natural'' choice for their derivatives should be.
\end{rem}

Before stating a progress theorem for $\lambda_\epsilon$, we must point out one
small subtlety, as the definition of reduction of unrestricted terms depends on
the particular representation chosen for the term. For example, the terms 
$((\lambda x . x) + 0)\ 0$ and $(\lambda x . x)\ 0$ are equivalent, but the first
one contains no $\beta$-redexes, whereas the second one reduces to $0$ in one
step. We can prove that progress holds for canonical terms, however, as those
are ``maximally reducible''.

\begin{defi}
  A canonical term $T$ is a \textbf{canonical value} whenever it is of the form
  \[ T = \sum_{i=1}^i \varepsilon^{k_i}(\lambda x_i . t_i)
  \]
\end{defi}
\begin{thm}[Progress]
  Whenever a canonical term $T$ admits a typing derivation $\ts T :
  \tau$, then either $T$ is a canonical value or there is some term $t'$ with $T
  \reducesTo t'$.
\end{thm}
\begin{proof}
  The proof proceeds by induction on the structure of $T$.
  \begin{itemize}
    \item When $T = 0$ then $T$ is trivially a canonical value.
    \item When $T = \epsilon^k s^\b$, then $s^\b$ has the form $\lambda x .
    e^\b$ or $\D(e^\b) \cdot u^\b$. In the first case $T$ is already a 
    canonical value. In the second case, note that the term $e^\b$ is itself a
    canonical term and a strict subterm of $T$. By inversion of the typing
    rules, we have that $\ts e^\b : \tau$ (note that the type of a differential
    application is the same as the type of its body, unlike the case of standard
    application). Hence either $e^\b reduces$ (and therefore so does $T$),
    or it is a canonical value, i.e. $e^\b$ is of the form $\lambda x . w^\b$.
    But in the last case then $T = \epsilon^k \D(\lambda x . w^\b) \cdot s^\b$
    and therefore $T \reducesTo \epsilon^k \lambda x . \diffS{w^\b}{x}{s^\b}$.
    \item When $T = T_1 + T_2$, then either both $T_1, T_2$ are canonical values,
    and then so is $T$, or one of $T_1, T_2$ reduces, in which case so does $T$.
    \qedhere
  \end{itemize}
\end{proof}

\begin{defi}
  We extend typing judgements to well-formed terms by setting $\Gamma \mathrel{\ul{\ts}} \ul{t} 
  : \tau$ whenever $\Gamma \ts \can{t} : \tau$.
\end{defi}

\begin{cor}[Subject reduction for well-formed terms]
  \label{cor:well-formed-subject-reduction}
  Whenever $\Gamma \wfTs \ul{t} : \tau$ and $\ul{t} \wfReducesTo
  \ul{t'}$, then $\Gamma \wfTs \ul{t'} : \tau$.
\end{cor}
\begin{proof}
  Since $\ul{t} \wfReducesTo \ul{t'}$, there is some canonical form $T =
  \can{\ul{t}}$ such that $T \reducesTo t'' \diffEq t'$, and furthermore $\Gamma
  \ts T : \tau$. By definition of 
  Theorem~\ref{thm:subject-reduction}, we have that $\Gamma \ts t'' : \tau$ and
  therefore by Proposition~\ref{prop:typing-canonical} $\Gamma \ts \can{t''} :
  \tau$, from which it follows that $\Gamma \wfTs \ul{t'} : \tau$.
\end{proof}

\begin{cor}[Progress for well-formed terms]
  \label{cor:well-formed-progress}
  Whenever $\Gamma \wfTs \ul{t} : \tau$ then either $\ul{t} \wfReducesTo
  \ul{t'}$ or every canonical form $\can{\ul{t}}$ is a canonical value. 
\end{cor}

\subsection{Strong Normalisation}

With our typing rules in place, we set out to show that $\lambda_\epsilon$ is
strongly normalising. Our proof follows the structure of Ehrhard and Regnier's
\cite{ehrhard2003differential} and Vaux's\cite{vaux2006lambda}, which use an
adaptation of the well-known argument by reducibility candidates. Our proof will
be somewhat simpler, however, due to two main reasons: first, we are not
concerning ourselves with terms with coefficients on some general rig; and
second, we have defined unrestricted and canonical terms as inductive types, and
so we can freely use induction on the syntax of our terms. We will need some
auxiliary results, which we prove now.

\begin{lem}
  Given an unrestricted term $t$, there are only finitely many terms $t'$ such that
  $t \reducesTo t'$.
\end{lem}
\begin{proof}
  Since our reduction relation is defined by simple induction on the syntax of
  $t$, it suffices to observe that any term $t$ will only contain a finite number
  of applications where reduction may take place.
\end{proof}

\begin{lem}
  Given a well-formed term $\ul{t}$, there are only finitely many canonical
  terms $T$ such that $T \diffEq \ul{t}$.
\end{lem}
\begin{proof}
  By Theorem~\ref{thm:canonicity}, we know that any two canonical forms for
  $\ul{t}$ must be permutatively equivalent. But any term has a finite number of
  permutative equivalence classes, hence a term $\ul{t}$ only has finitely many
  canonical forms.
\end{proof}
\DeclarePairedDelimiter\abs{\lvert}{\rvert}%

As a corollary of the two previous results, whenever a well-formed term $\ul{t}$
is strongly normalising, by K\"onig's lemma there is a longest sequence of
well-formed terms $\ul{t} = \ul{t_1}, \ul{t_2}, \ldots, \ul{t_n}$ such that
$\can{t_i} \reducesTo t'_i \diffEq t_{i+1}$. We write $\abs{\ul{t}}$ to indicate
the length of this sequence. The following result is then immediate.

\begin{lem}
  Whenever $\ul{t}$ is strongly normalising and $\ul{t} \reducesTo \ul{t'}$, we
  have $\abs{\ul{t}} > \abs{\ul{t'}}$.
\end{lem}

\begin{lem}
  \label{lem:sum-strongly-normalising}
  A term $\ul{s + t}$ is strongly normalising if and only if
  $\ul{s}, \ul{t}$ are strongly normalising.
\end{lem}
\begin{proof}
  The proof in the first direction proceeds by induction on $\abs{s + t}$.

  Suppose $\ul{s + t}$ is strongly normalising.  and suppose $S, T$
  are canonical forms for $\ul{s}, \ul{t}$ respectively. Then $S + T$ is
  a canonical form for $\ul{s + t}$ (up to associativity of addition), and
  any sequence of reductions 

  If $\abs{\ul{s + t}} = 0$ it follows that its canonical form $S +
  T$ does not reduce, and hence neither do the separate $S + T$, thus
  $\ul{s}, \ul{t}$ are normal forms.

  On the other hand, suppose $S \reducesTo s', T \reducesTo t'$ then
  $\ul{s + t} \wfReducesTo \ul{s' + t'}$ which is therefore
  strongly normalising, with $\abs{\ul{s' + t'}} < \abs{t}$. Hence, by
  induction, $\ul{s'}$ and $\ul{t'}$ are strongly normalising for any
  choices of canonical forms $S, T$ and reducts $s', t'$.

  In the opposite direction, the proof follows similarly by induction on
  $\abs{s} + \abs{t}$. The base case is equally trivial. For the inductive step,
  since any canonical form for $\ul{s + t}$ is (up to commutativity of addition)
  of the form $S + T$, with $S, T$ being canonical forms for $\ul{s}, \ul{t}$
  respectively, it follows that if $\ul{s + t} \wfReducesTo \ul{e}$. Without
  loss of generality, we assume that $S \wfReducesTo S'$ with $\can{\ul{e}} = S'
  + T$. Then $\ul{s} \wfReducesTo \ul{S'}$ and $\ul{e} = \ul{S' + t}$. Now
  $\abs{\ul{S'}} < \abs{\ul{s}}$, and so we apply the induction hypothesis and
  obtain that $\ul{e}$ must be strongly normalising, and therefore so is $\ul{s
  + t}$.
\end{proof}

\begin{lem}
  \label{lem:epsilon-strongly-normalising}
  A term $\ul{\epsilon s}$ is strongly normalising if and only if $\ul{s}$ is.
\end{lem}

\newcommand{\rc}[1]{\mathcal{R}_{#1}}

\begin{defi}
  For every
  type $\tau$ we introduce a set $\rc{\tau}$ of well-formed terms of type
  $\tau$. We do so by induction on $\tau$.
  \begin{itemize}
    \item Whenever $\tau = \mathbf{t}$ is a primitive type, $\ul{s} \in
    \rc{\mathbf{t}}$ if and only if $\ul{s}$ is strongly normalising.
    \item Whenever $\tau = \sigma_1 \Ra \sigma_2$, $\ul{s} \in \rc{\sigma_1 \Ra
    \sigma_2}$ if and only if for any additive term $\ul{t^*}
    \in \rc{\sigma_1}$ and for any sequence
    $\ul{v^\b_1}, \ldots, \ul{v^\b_n}$ of 
    basic terms $\ul{v^\b_i} \in \rc{\sigma_1}$ 
    of length $n \geq 0$
    we have $\ul{\pa{\D^n(s) \cdot (v^\b_1, \ldots, v^\b_i)}\ t^*} \in \rc{\sigma_2}$ 
  \end{itemize}

  If $\ul{t} \in \rc{\tau}$ we will often just say that $\ul{t}$
  is \textbf{reducible} if the choice of $\tau$ is clear from the context.
\end{defi}

\begin{lem}
  \label{lem:rc-renaming}
  Whenever $\ul{t} \in \rc{\tau}$, then for any two distinct variables $x, y$
  the renaming $\ul{\subst{t}{x}{y}}$ is also in $\rc{\tau}$.
\end{lem}
\begin{proof}
  Straightforward induction on $\tau$.
\end{proof}

\begin{lem}
  \label{lem:rc-strongly-normalising}
  Whenever $\ul{t} \in \rc{\tau}$, then $\ul{t}$ is strongly normalising.
\end{lem}
\begin{proof}
  By induction on $\tau$. When $\tau$ is a primitive type, the result follows
  trivially. Let $\tau = \sigma_1 \Ra \sigma_2$ and $\ul{t} \in
  \rc{\tau}$. By the induction hypothesis we know that for all $\ul{u} \in
  \rc{\sigma_2}$ the application $\ul{t\ u}$ is strongly normalising.

  Now suppose some canonical form of $\ul{t}$ reduces, that is, $T \reducesTo
  t'$. Since $\can{t\ u^*} = \ap{\can{t}}{\pri{u^*}} = \ap{T}{{\can{u^*}}}$, it
  follows that $\can{t\ u} \reducesTo^+ \ap{t'}{\pri{u}}$. Hence if there were
  any infinite sequence of reductions starting from $\ul{t}$, so would there be
  an infinite sequence of reductions starting from $\ul{t\ u}$. Since $\ul{t\ u}$
  is strongly normalising, this must be impossible and so $\ul{t}$ must be
  strongly normalising as well.
\end{proof}

\begin{lem}
  \label{lem:rc-sum}
  Whenever $\ul{s}, \ul{t} \in \rc{\tau}$, then both $\ul{s + t},
  \ul{\varepsilon{s}}$ are in $\rc{\tau}$. Conversely, whenever $\ul{s + t}$
  is in $\rc{\tau}$ then so are $\ul{s}, \ul{t}$.
\end{lem}
\begin{proof}
  When $\tau$ is a primitive type the proof is a straightforward corollary of 
  Lemmas~\ref{lem:sum-strongly-normalising}
  and~\ref{lem:epsilon-strongly-normalising}.

  When $\tau = \sigma_1 \Ra \sigma_2$, consider an additive term $\ul{u^*} \in
  \rc{\sigma_1}$. We ask whether the application $\ul{(s + t)\ u^*}$ is in
  $\rc{\sigma_2}$. But note that $\can{(s + t)\ u^*}$ is equal to $\can{s\ u^*} +
  \can{t\ u^*}$ (modulo commutativity of the sum) and therefore $\ul{(s + t)\ u^*}
  = \ul{s\ u^*} + \ul{t\ u^*}$. Since $\ul{s\ u^*}$ and $\ul{t\ u^*}$ are both in
  $\rc{\sigma_2}$, it follows by the induction hypothesis that so is $\ul{(s +
  t)\ u^*}$. The same reasoning shows that $\ul{(\D(s + t) \cdot v^\b)\ u^*}$ is
  in $\rc{\sigma_2}$.
  
  The proof for $\varepsilon$ follows by a simpler but otherwise identical
  procedure.

  On the opposite direction, consider a sum $\ul{s + t} \in \rc{\tau}$. When
  $\tau$ is a primitive type then $\ul{s + t}$ is strongly normalising and
  therefore so are $\ul{s}, \ul{t}$, hence they are regular. On the other hand,
  if $\tau = \sigma_1 \Ra \sigma_2$, we have that for any $e^* \in
  \rc{\sigma_1}$ the reducible $\ul{(s + t)\ e^*}$ is equal to $\ul{(s\ e^*) +
  (t\ e^*)}$. By the induction hypothesis, both $(s\ e^*)$ and $(t\ e^*)$ are
  regular. A similar argument proves that differential applications of $s$ and
  $t$ are also regular, thus $s, t \in \rc{\sigma_1 \Ra \sigma_2}$.
\end{proof}

\begin{lem}
  \label{lem:rc-d}
  Whenever $\ul{s} \in \rc{\sigma \Ra \tau}$ and $\ul{t} \in \rc{\sigma}$ then
  $\ul{\D (s) \cdot t} \in \rc{\sigma \Ra \tau}$.
\end{lem}
\begin{proof}
  Pick a canonical form $T$ of $t$. The proof proceeds by induction on the
  number of summands of $T$. If $T = 0$ or $T = \epsilon^k t^\b$ the result
  follows directly by definition of $\rc{\sigma \Ra \tau}$.

  Now suppose $T = \epsilon^k t^\b + T'$. Then
  \begin{align*}
    \D(s) \cdot t &\diffEq \D(s) \cdot (t^\b + T')\\
    &\diffEq \D(s) \cdot t^\b + \D(s) \cdot T'
    + \varepsilon \D(\D(s) \cdot T') \cdot t^\b
  \end{align*}
  Now $\D(s) \cdot t^\b$ is evidently reducible (as $s$ is reducible by
  hypothesis and, by Lemma~\ref{lem:rc-sum}, $t^\b$ is also reducible), and by
  the induction hypothesis so is $\D(s) \cdot T'$, from which also follows that
  $\D(\D(s) \cdot T') \cdot t^\b$ is reducible as well. By
  Lemma~\ref{lem:rc-sum} it follows that $\ul{t}$ is reducible.
\end{proof}

\begin{cor}
  A well-formed term $\ul{t}$ is in $\rc{\tau}$ if and only if some canonical
  form $T = \can{\ul{t}}$ is of the form $\sum_{i=1}^n\varepsilon^{k_i}t^\b_i$
  with $\ul{t^\b_i} \in \rc{\tau}$ for each $1 \le i \le n$.
\end{cor}

\begin{lem}
  Whenever $\ul{t} \in \rc{\tau}, \ul{t} \wfReducesTo^+ \ul{t'}$, then $\ul{t'}
  \in \rc{\tau}$.
\end{lem}
\begin{proof}
  We proceed by induction on $\tau$.
  When $\tau$ is a primitive type, we have that $\ul{t}$ is strongly normalising
  and therefore so is $\ul{t'}$, hence $\ul{t'} \in \rc{\tau}$.

  When $\tau = \sigma_1 \Ra \sigma_2$, we pick some additive reducible term
  $\ul{e^*}$ and a sequence of reducible basic terms $\ul{u^\b_1}, \ldots,
  \ul{u^\b_k}$, and establish that $\pa{\D^k(t') \cdot(u^\b_1, \ldots,
  u^\b_k)}\ e^*$ is reducible. But this is immediate: since $\ul{t}$ reduces to
  $\ul{t'}$, then so does $\ul{\D^k(t) \cdot (u^\b_1, \ldots, u^\b_k)\ e^*}$
  reduce to $\ul{\D^k(t') \cdot (u^\b_1, \ldots, u^\b_k)\ e^*}$ and, by induction,
  $\ul{\D^k(t') \cdot(u^\b_1, \ldots, u^\b_k)\ e^*}$ is reducible.
\end{proof}

\begin{defi}
  A basic term $t^\b$ is \textbf{neutral} whenever it is not a
  $\lambda$-abstraction. In other words, a basic term is neutral whenever it is
  of the form $x, (s\ t)$ or $\D(s)\cdot u$.

  A canonical term $T$ is neutral whenever it is of the form
  $\sum_{i=1}^n\varepsilon^{k_i} s_i^\b$, where each of the $s_i^\b$ are
  neutral. In particular, $0$ is a neutral term.

  A well-formed term $\ul{t}$ is neutral whenever some canonical form (and
  therefore all of its canonical forms) is neutral.
\end{defi}

\begin{lem}
  \label{lem:rc-neutral}
  Whenever $\ul{t}$ is neutral and every $\ul{t'}$ such that $\ul{t}
  \wfReducesTo^+ \ul{t'}$ is in $\rc{\tau}$, then so is $\ul{t}$.
\end{lem}
\begin{proof}
  When $\tau$ is a primitive type the proof is immediate, as $\ul{t'} \in
  \rc{\tau}$ implies $\ul{t'}$ is strongly normalising and therefore so is
  $\ul{t}$.

  When $\tau = \sigma_1 \Ra \sigma_2$, we show the reasoning for standard
  application first. We select arbitrary $\ul{e^*}, \ul{u^\b_1}, \ldots,
  \ul{u^\b_k} \in \rc{\sigma_1}$ and show that whenever $\pa{\D^k(t) \cdot
  (u^\b_1, \ldots, u^\b_k)}\ e^*$ reduces then it reduces to a reducible term
  (hence our desired result will follow by induction on $\tau$). 

  We prove this property by induction on $Q \defeq \abs{\ul{e^*}} +
  \abs{\ul{u^\b_1}} + \ldots + \abs{\ul{u^\b_k}}$ which is well-defined since,
  by hypothesis, all of the involved terms are strongly normalising. 
  
  When $Q = 0$ then all of our chosen terms are normal, and so, since $t$ is
  neutral, if $\ul{\pa{\D^k(t)\cdot(u^\b_1, \ldots, u^\b_k)}\ e^*}$ reduces it
  must be that $\ul{t} \wfReducesTo^+ \ul{t'}$. By hypothesis, $\ul{t'} \in
  \rc{\sigma_1 \Ra \sigma_2}$ and therefore $\ul{\pa{\D^k(t')\cdot(u^\b_1,
  \ldots, u^\b_k)}\ e^*}$ is reducible.

  When $Q > 0$, then a reduction may occur in $t$, in which case we apply the
  previous reasoning, or in one of the applied terms, in which case we apply the
  induction hypothesis on $Q$.
  induction hypothesis.
\end{proof}

\begin{lem}
  \label{lem:rc-lambda}
  If, for all $\ul{t^*} \in \rc{\sigma_1}$ where $x$ does not appear free, the
  term $\ul{\subst{s}{x}{t^*}}$ is in $\rc{\sigma_2}$ and, for all $\ul{u^\b}$
  where $x$ does not appear free, the term
  $\ul{\subst{\pa{\diffS{s}{x}{u^\b}}}{x}{t^*}}$ is in $\rc{\sigma_2}$, then the
  term $\ul{\lambda x . s}$ is in $\rc{\sigma_1 \Ra \sigma_2}$.
\end{lem}
\begin{proof}
  As a corollary of Lemma~\ref{lem:rc-sum}, it suffices to check the case
  when $s$ is some basic term $s = s^\b$. 
  
  Pick any variable $y \neq$ z. Since the variable $y$ itself is an additive
  term in $\rc{\sigma_1}$, by hypothesis we have $\ul{\subst{s}{x}{y}} \in
  \rc{\sigma_2}$. But since reducible terms are closed under renaming as per
  Lemma~\ref{lem:rc-renaming}, this means that so is $\ul{s}$.

  Now consider an arbitrary $\ul{e^*} \in \rc{\sigma_1}$. We show that the
  application $\ul{(\lambda x . s^\b)\ e^*}$ is in $\rc{\sigma_2}$. As it is a
  neutral term, by Lemma~\ref{lem:rc-neutral} it suffices to prove that every
  one-step reduct of $(\lambda x . s^\b)\ e^*$ is in $\rc{\sigma_2}$. We do so by
  induction on $\abs{s^\b} + \abs{e^*}$. The term $(\lambda x . s^\b)\ e^*$
  reduces to one of:
  \begin{itemize}
    \item $\subst{s^\b}{x}{e^*}$, which by hypothesis is a representative of a
    term in $\rc{\sigma_2}$.

    \item $(\lambda x . s')\ e^*$ with $s^\b \reducesTo s'$. Then $s' \in
    \rc{\sigma_2}$ and since $\abs{s'} < \abs{s^\b}$ we apply our induction on
    $\abs{s^\b}$ to obtain that $(\lambda x . s')\ e^* \in \rc{\sigma_2}$.

    \item $(\lambda x . s^\b)\ e'$ with $e^* \reducesTo e'$. Then $e' \in
    \rc{\sigma_1}$ and since $\abs{e'} < \abs{e^*}$ we apply our induction on
    $\abs{e^*}$ to obtain that $(\lambda x . s^\b)\ e' \in \rc{\sigma_2}$.
  \end{itemize}

  By a similar argument we can show that $\ul{
    \pa{\D^k(\lambda x . s^\b)\cdot(u^\b_1, \ldots, u^\b_k)}\ e^*
  }$ is reducible, applying Lemma~\ref{lem:rc-neutral} and using induction
  in $\abs{e^*} + \abs{u^\b_1} + \ldots + \abs{u^\b_k}$. 
\end{proof}

\begin{thm}
  Consider a well-formed term $\ul{t}$ which admits a typing of the form $x_1 :
  \sigma_1, \ldots, x_n : \sigma_n \wfTs \ul{t} : \tau$ and assume given the
  following data:
  \begin{itemize}
    \item A sequence of basic terms 
      $\ul{d^\b_1} \in \rc{\sigma_1}, \ldots, \ul{d^\b_n} \in \rc{\sigma_n}$.
    \item An arbitrary sequence of indices 
      $i_1, \ldots, i_k \in \{1, \ldots, n\}$ (possibly with repetitions).
    \item A sequence of additive terms 
      $\ul{s^*_1} \in \rc{\sigma_{i_1}}, \ldots, \ul{s^*_k} \in \rc{\sigma_{i_k}}$.
  \end{itemize}
  such that none of the variables $x_1, \ldots, x_i$ appear free in the $d^\b_i,
  s^*_i$. Then the term 
  \[
    \ul{t'} = \ul{
      \pa{
        \frac{\partial^k t}{\partial(x_{i_1}, \ldots, x_{i_k})}
        (d^\b_1, \ldots, d^\b_k)}
      \left[s^*_1, \ldots, s^*_n / x_1, \ldots, x_n\right]}
  \]
  is in $\rc{\tau}$.
\end{thm}
\begin{proof}
  \newcommand{\thediff}[1]{
      \pa{\partial^k t / \partial x_{i_1} \ldots \partial x_{i_k}}
      \pa{d^\b_1, \ldots, d^\b_k}
  }
  \newcommand{\thesubs}[1]{
    {#1}\brak{s^*_1, \ldots, s^*_n / x_1, \ldots, x_n}
  }
  \newcommand{\thesubss}[1]{
    \subst{#1}{\overline{x_i}}{\overline{s^*_i}}
  }
  \newcommand{\thesd  }[1]{
    \thesubs{\pa{\thediff{#1}}}
  }
  Throughout the proof we will write $\overline{x_i}, \overline{s^*_i},
  \overline{d^\b_i}$ as a shorthand for the corresponding sequences $x_1,
  \ldots, x_n$, etc.

  By definition of $\wfTs$, we know that there is some canonical form $T$ (in
  fact any canonical form) of $\ul{t}$ such that $x_1 : \sigma_1, \ldots, x_n :
  \sigma_n \ts T : \tau$. We prove our property holds by induction on this
  typing derivation. Furthermore, by Lemma~\ref{lem:rc-sum}, it suffices to
  consider the case when $T$ is in fact some basic term $t^\b$. 
  We proceed now by case analysis on the last rule of the typing derivation:

  \begin{itemize}
    \item $t^\b = x_i$ (and therefore $\tau = \sigma_i$)
    
      If the sequence of indices $i_1, \ldots, i_k$ is empty then the
      substitution $\ul{t'}$ is exactly equal to $\ul{s^*_i}$ and therefore
      $\ul{t'} \in \rc{\sigma_i}$.

      If the sequence of indices $i_1, \ldots, i_k$ is exactly the sequence
      containing only $i$ then since $x_i$ does not appear free in the
      substituted term $d^\b$ then $t'$ is differentially equivalent to
      $\pa{\diffS{x_i}{x_i}{d^\b_i}}$ and therefore $\ul{t'} = \ul{d^\b_i} \in
      \rc{\sigma_i}$.

      If the list of indices contains two or more indices, or does not contain
      $i$, then $t' \diffEq 0$ and therefore $\ul{t'}$ is trivially in $\in
      \rc{\sigma_i}$ (either the derivative)

    \item $t^\b = \D(s^\b) \cdot e^\b$
    
      Applying Lemma~\ref{lem:differential-substitution-dapp}, we know that
      the term \[ \thesd{t^\b} \] is equivalent to a sum of terms of the
      form
      \[
        \pa{\epsilon^z \D^l \pa{\thesubss{v}}
        \cdot \pa{\thesubss{w_1}, \ldots, \thesubss{w_l}}}
      \]
      Again by Lemmas~\ref{lem:rc-sum} and~\ref{lem:rc-d}, it suffices to show
      that each of the $\thesubss{v}$, $\thesubss{w_j}$ are reducible. But
      by Lemma~\ref{lem:differential-substitution-app}, we know that $v$ has
      the form
      \[
        \frac{\partial^p s^\b}{\partial (x_{j_1}, \ldots, x_{j_m})}\pa{
          d^\b_{j_1}, \ldots,
          d^\b_{j_m}
        }
      \]

      Since $s^\b$ is a subterm of $t^\b$ its typing derivation is therefore a
      sub-derivation of the one for $t^\b$. We apply the induction
      hypothesis, obtaining that $\ul{\thesubss{v}}$ is reducible
      (as each of the $d^\b_i$ are reducible). By a similar argument,
      each of the $\ul{\thesubss{w_j}}$ are reducible as well, and therefore
      so is $\ul{t^\b}$.

    \item $t^\b = (s^\b\ e^*)$
    
      Applying Lemma~\ref{lem:differential-substitution-dapp}, we know that
      the term \[ \thesd{t^\b} \] is equivalent to a sum of terms of the
      form
      \[
        \pa{\epsilon^z \D^l \pa{\thesubss{v}}
        \cdot \pa{\thesubss{w_1}, \ldots, \thesubss{w_l}}}\ \pa{\thesubss{e^*}}
      \]
      Again by Lemma~\ref{lem:rc-sum} it suffices to show that every such term
      is reducible. First, by Lemma~\ref{lem:differential-substitution-app}, we
      know that $v$ has the form
      \[
        \frac{\partial^p s^\b}{\partial (x_{j_1}, \ldots, x_{j_m})}\pa{
          d^\b_{j_1}, \ldots,
          d^\b_{j_m}
        }
      \]
      Since $s^\b$ is a subterm of $t^\b$ its typing derivation is therefore a
      sub-derivation of the one for $t^\b$. We apply the induction hypothesis,
      obtaining that $\ul{\thesubss{v}}$ is reducible (as each of the $d^\b_i$
      are reducible). By a similar argument, each of the $\ul{\thesubss{w_j}}$
      are reducible as well, as is $\ul{\thesubss{e^*}}$. Thus, by
      Lemma~\ref{lem:rc-d}, the entire differential application is reducible.
      
      But since $t^\b$ is an application of the reducible term
      \[ \epsilon^z \D^l \pa{\thesubss{v}} \cdot \pa{\thesubss{w_1}, \ldots,
      \thesubs{w_l}} \] to the reducible term $\thesubs{e^*}$, it follows then
      that $\ul{t^\b}$ is itself reducible.

    \item $t^\b = \lambda y . s^\b$
    
      Pick fresh variables $z_1, \ldots, z_n$. By the induction hypothesis, we
      know that both of the following substitutions are reducible:
      \begin{gather*}
        \subst{\pa{t^\b}}{x_1, \ldots, x_n, y}{z_1, \ldots, z_n, e^*}\\
        \subst{\pa{\diffS{t^\b}{y}{d^\b}}}{x_1, \ldots, x_n, y}{z_1, \ldots, z_n, e^*}
      \end{gather*}
      But since reducible terms are closed under renaming, then the terms
      $\subst{t^\b}{y}{e^*}$ and $\subst{\pa{\diffS{t^\b}{y}{d^\b}}}{y}{e^*}$
      are also reducible. Hence, by Lemma~\ref{lem:rc-lambda}, the term 
      $\lambda y . s^\b$ is reducible.\qedhere
  \end{itemize}
\end{proof}

\begin{cor}[Strong normalisation]
  Whenever a closed well-formed term is
  typable with type $\wfTs \ul{t} : \tau$, it is strongly normalising.
\end{cor}
\begin{proof}
  By the previous result, we know that any such term satisfies $\ul{t} \in
  \rc{\tau}$, and therefore $\ul{t}$ is strongly normalising.
\end{proof}

\section{Semantics}
\label{sec:semantics}

It is a well-known result that the differential $\lambda$-calculus can be
soundly interpreted in any differential $\lambda$-category, that is to say, any
Cartesian differential category where differentiation ``commutes with''
abstraction (in the sense of \cite[Definition~4.4]{bucciarelli2010categorical}).

The exact same result holds for the difference $\lambda$-calculus and difference
$\lambda$-categories. In what follows we will consider a fixed difference
$\lambda$-category $\cat{C}$, and proceed to define interpretations for the
types, contexts and terms of the simply-typed $\lambda_\epsilon$-calculus.

\begin{defi}
  Given a $\mathbf{t}$-indexed family of objects $O_\mathbf{t}$, we define the
  interpretation $\sem{\tau}$ of a type $\tau$ by induction on its structure by
  setting $\sem{\mathbf{t}} \defeq O_\mathbf{t}$, $\sem{\sigma \Ra \tau} \defeq
  \sem{\sigma} \Ra \sem{\tau}$. We lift the interpretation of types to contexts
  in the usual way. Or, more formally, we have: $\sem{\cdot} \defeq \terminal$,
  $\sem{\Gamma, x : \tau} \defeq \sem{\Gamma} \times \sem{\tau}$.
\end{defi}

\begin{defi}
  Given a well-typed unrestricted $\lambda_\epsilon$-term
  $\Gamma \ts t : \tau$, we define its interpretation $\sem{t} : \sem{\Gamma}
  \ra \sem{\tau}$ inductively as in Figure~\ref{fig:semantics} below. When
  $\Gamma$ and $\tau$ are irrelevant or can be inferred from the context, we
  will simply write $\sem{t}$.
  \begin{figure}[ht]
  \[\def\arraystretch{1.3}
    \begin{array}{rccl}
      \sem{(x_i : \tau_i)_{i = 1}^n \ts x_k : \tau_k} 
      &\defeq& 
      \pi_2 \circ \pi_1^{n - k} 
      &: \prod_{i=1}^n\sem{\tau_i} \to \sem{\tau_k}
      \\
      \sem{\Gamma \ts 0 : \tau} &\defeq&
      0 
      &: \sem{\Gamma} \to \sem{\tau}
      \\
      \sem{\Gamma \ts s + t : \tau} &\defeq&
      \sem{s} + \sem{t} 
      &: \sem{\Gamma} \to \sem{\tau}
      \\
      \sem{\Gamma \ts \epsilon t : \tau} &\defeq&
      \epsilon \sem{t} 
      &: \sem{\Gamma} \to \sem{\tau}
      \\
      \sem{\Gamma \ts \lambda x . t : \sigma \Ra \tau} &\defeq&
      \Lambda \sem{t} 
      &:\sem{\Gamma} \to \sem{\sigma} \Ra \sem{\tau}
      \\
      \sem{\Gamma \ts (s\ t) : \tau} &\defeq&
      \A{ev} \circ \pair{\sem{s}}{\sem{t}} 
      &:\sem{\Gamma} \to \sem{\tau}
      \\
      \sem{\Gamma \ts \D(s) \cdot t : \sigma \Ra \tau} &\defeq&
      \Lambda(\dd{\Lambda^-(\sem{s})} \circ \pair{\Id{}}
      {\pair{0}{\sem{t} \circ \pi_1}})
      &: \sem{\Gamma} \to \sem{\tau}
    \end{array}
  \]
  \caption{Interpreting $\lambda_\epsilon$ in $\cat{C}$}
  \label{fig:semantics}
  \end{figure}
\end{defi}

\begin{lem}
  \[
    \dd{\Lambda^-f} \circ \pair{\Id{}}{\pair{0}{g \circ \pi_1}}
    = \dd{\A{ev}} \circ \pair{}{} \circ
  \]
\end{lem}

\newcommand{\semEq}[0]{\sim_{\llbracket\rrbracket}}
\begin{lem}
  Define the relation $\semEq \subseteq \Lambda_\epsilon$ by letting $s \semEq
  t$ whenever there exist $\Gamma, \tau$ such that
  $\sem{\Gamma \ts s : \tau} = \sem{\Gamma \ts t : \tau}$. Then the relation
  $\semEq{}$ is contextual.
\end{lem}

\begin{thm}
  Whenever $s \diffEq t$ are equivalent unrestricted terms that admit typing
  derivations $\Gamma \ts s : \tau, \Gamma \ts t : \tau$, then their
  interpretations are identical, that is to say:
  \[
    \sem{\Gamma \ts s : \tau} = \sem{\Gamma \ts t : \tau}
  \]
\end{thm}
\begin{proof}
  The result can be equivalently stated as ``$\semEq$ contains $\diffEq$''.
  Since $\diffEq$ is defined as the least contextual equivalence relation that
  contains $\diffEq^1$, and $\semEq$ is both contextual and an equivalence
  relation, it suffices to prove that it contains $\diffEq^1$.

  We examine the rules in Figure~\ref{fig:differential-equivalence} and show
  that they all hold, by checking that each rule verifies $\sem{\LHS} =
  \sem{\RHS}$. The first and second blocks are trivial, as the rules correspond
  precisely to stating that $\cat{C}$ is a Cartesian closed left-additive
  category with an infinitesimal extension which is compatible with the
  Cartesian structure. The third block is a trivial consequence of \CdCAx{1} and
  so we will omit it as well, but we give explicit proofs for some
  of the properties of the fourth block.

  First we show a general equivalence between syntactic and semantic second
  derivatives which will simplify the task considerably.
  \begin{align*}
    &\hspace{-10pt}\Lambda^-(\sem{\D(\D(s) \cdot u) \cdot v})
    \\
    &= \dd{\Lambda^-(\sem{\D(s)\cdot u})} \circ \pair{\Id{}}{\pair{0}{\sem{v} \circ \pi_1}}
    \\
    &= 
    \dd[^2]{\Lambda^-(\sem{s})} \circ \pair{\Id{}}{\pair{0}{\sem{u} \circ \pi_1}}]
    \pair{\Id{}}{\pair{0}{\sem{v} \circ \pi_1}}
    \\
    &=
    \dd[^2]{\Lambda^-(\sem{s})}
    \circ 
    \pair{\pair{\Id{}}{\pair{0}{\sem{u} \circ \pi_1}} \circ \pi_1}
    {\dd{\pair{\Id{}}{\pair{0}{\sem{u}\circ \pi_1}}}}
    \circ
    \pair{\Id{}}{\pair{0}{\sem{v}\circ \pi_1}}
    \\
    &=
    \dd[^2]{\Lambda^-(\sem{s})}
    \\
    &\quad
    \circ 
    \pair
    {\pair{\Id{}}{\pair{0}{\sem{u} \circ \pi_1}}}
    {\pair{\dd{\Id{}} \circ \pair{\Id{}}{\pair{0}{\sem{v} \circ \pi_1}}}
          {\dd{\sem{v} \circ \pi_1} \circ \pair{\Id{}}{\pair{0}{\sem{v} \circ \pi_1}}}}
    \\
    &= 
    \dd[^2]{\Lambda^-(\sem{s})}
    \circ 
    \pair
    {\pair{\Id{}}{\pair{0}{\sem{u} \circ \pi_1}}}
    {\pair{\pair{0}{\sem{v} \circ \pi_1}}
          {\dd{\sem{v}} \circ \pair{\pi_1}{0}}}
    \\
    &= 
    \dd[^2]{\Lambda^-(\sem{s})}
    \circ 
    \pair
    {\pair{\Id{}}{\pair{0}{\sem{u} \circ \pi_1}}}
    {\pair{\pair{0}{\sem{v} \circ \pi_1}}{0}}
  \end{align*}

  With the above, most of the conditions become trivial. For example, we can
  prove that regularity of the syntactic derivative follows from semantic
  regularity (that is to say, \CdCAx{2}) with the following calculation:
  \begin{itemize}
    \item $\D(s) \cdot (u + v) \diffEq^1 \D(s) \cdot u + \D(s) \cdot v
    + \epsilon(\D(\D(s) \cdot u)\cdot v)$
    \begin{align*}
      \sem{\D(s) \cdot (u + v)} &=
      \Lambda(\dd{f} \circ \pair{\Id{}}{\pair{0}{\sem{u} \circ \pi_1 + \sem{v} \circ \pi_1}})
      \\
      &=
      \Lambda(\dd{f} \circ \pair{\Id{}}{\pair{0}{\sem{u} \circ \pi_1} + \pair{0}{\sem{v} \circ \pi_1}})\\
      &=
      \Lambda\big[
                (\dd{f} \circ \pair{\Id{}}{\pair{0}{\sem{u} \circ \pi_1}})\\
      &\quad
        + (\dd{f} \circ \pair{\Id{} + \epsilon(\pair{0}{\sem{u} \circ \pi_1})}{\pair{0}{\sem{v} \circ \pi_1}})
      \big]
      \\
      &=
      \Lambda\brak{
                \dd{f} \circ \pair{\Id{}}{\pair{0}{\sem{u} \circ \pi_1}} }
      \\&\quad
        + \Lambda \brak{(\dd{f} \circ \pair{\Id{}}{\pair{0}{\sem{v} \circ \pi_1}})}
      \\&\quad
        + \epsilon\Lambda\brak{\dd[^2]{f} \circ 
        \pair{\pair{\Id{}}{\pair{0}{\sem{v} \circ \pi_1}}}{\pair{\pair{0}{\sem{u} \circ \pi_1}}{0}}
        }
      \big]
      \\
      &=
      \Lambda\brak{
                \dd{f} \circ \pair{\Id{}}{\pair{0}{\sem{u} \circ \pi_1}} }
      \\&\quad
        + \Lambda \brak{(\dd{f} \circ \pair{\Id{}}{\pair{0}{\sem{v} \circ \pi_1}})}
      \\&\quad
        + \epsilon\Lambda\brak{\dd[^2]{f} \circ 
        \pair{\pair{\Id{}}{\pair{0}{\sem{u} \circ \pi_1}}}{\pair{\pair{0}{\sem{v} \circ \pi_1}}{0}}
        }
      \big]
      \\
      &= \sem{\D(s) \cdot u)} + \sem{\D(s) \cdot v} + \sem{\varepsilon(\D (\D(s) \cdot u) \cdot v)}
    \end{align*}
  \end{itemize}
  Most of the other conditions follow from similar arguments. For example,
  commutativity of the syntactic second derivative follows from commutativity of
  the semantic second derivative. The third and fifth conditions, dealing with
  infinitesimal extensions, may seem harder to prove, but they are both
  corollaries of axiom~\CdCAx{6}, as we showed in Lemma~\ref{lem:d-epsilon}.
  It remains to show
  that the syntactic derivative condition holds; this is not hard, but we do
  it explicitly as the derivative condition is such a central notion.

  \begin{itemize}
    \item $s\ (t + \varepsilon e) \diffEq^1 (s\ t) + \varepsilon((\D(s) \cdot
    e)\ t)$
    \begin{align*}
      &\hspace{-10pt}\sem{s\ (t + \varepsilon e)}\\
      &= \A{ev} \circ \pair{\sem{s}}{\sem{t} + \varepsilon\sem{e}}\\
      &= \Lambda^-\sem{s} \circ \pair{\Id{}}{\sem{t} + \varepsilon\sem{e}}\\
      &= \pa{\Lambda^-\sem{s} \circ \pair{\Id{}}{\sem{t}}}
        + \epsilon\pa{\dd{\Lambda^-\sem{s}} \circ 
        \pair{\pair{\Id{}}{\sem{t}}}{\pair{0}{\sem{e}}}}
      \\
      &= \sem{s\ t}
        + \epsilon\brak{
        \pa{\dd{\Lambda^-\sem{s}}
        \circ \pair{\Id{}}{\pair{0}{\sem{e} \circ \pi_1}}}
        \circ \pair{\Id{}}{\sem{t}}
      }
      \\
      &= \sem{s\ t}
        + \epsilon\brak{
          \A{ev} \circ \pair{
        \Lambda\pa{\dd{\Lambda^-\sem{s}}
        \circ \pair{\Id{}}{\pair{0}{\sem{e} \circ \pi_1}}}}
        {\sem{t}}
      }
      \\
      &= \sem{s\ t}
        + \epsilon\brak{
          \A{ev} \circ \pair{
            \sem{\D(s) \cdot e}
          }
        {\sem{t}}
      }
      \\
      &= \sem{s\ t} + \epsilon\sem{(\D(s) \cdot e)\ t}\qedhere
    \end{align*}
  \end{itemize}
\end{proof}

\begin{lem}
  Let $t$ be some unrestricted $\lambda_\epsilon$-term. The following properties
  hold:
  \begin{enumerate}[i.]
    \item If $\Gamma \ts t : \tau$ and $x$ does not appear in $\Gamma$ then 
    $\sem{\Gamma, x : \sigma \ts t : \tau} = \sem{\Gamma \ts t : \tau} \circ
    \pi_1$
    \item If $\Gamma, x : \sigma_1, y : \sigma_2 \ts t : \tau$ then
    $\sem{\Gamma, y:\sigma_2, x:\sigma_1 \ts t : \tau}
    = \sem{\Gamma, x : \sigma_1, y : \sigma_2 \ts t : \tau} \circ \A{sw}$
  \end{enumerate}
  The morphism $\A{sw}$ above is the obvious isomorphism between $(A \times B)
  \times C$ and $(A \times C) \times B$, which we can define explicitly by:
  \[
    \A{sw} \defeq \pair{\pair{\pi_{11}}{\pi_2}}{\pi_{21}}
    : (A \times B) \times C \ra (A \times C) \times B
  \]
\end{lem}

\begin{lem}
  Let $\Gamma, x : \tau \ts s : \sigma$, with $s$ some unrestricted
  $\lambda_\epsilon$-term. Then:
  
  \begin{enumerate}[i.]
  \item Whenever $\Gamma, x : \tau \ts t : \tau$, then
  $\sem{\subst{s}{x}{t}}_\Gamma = \sem{s}_{\Gamma, x : \tau} \circ
  \pair{\pi_1}{\sem{t}_{\Gamma, x : \tau}}$
  
  \item Whenever $\Gamma \ts t : \tau$, then $\sem{\diffS{s}{x}{t}}_{\Gamma,
  x : \tau} =
  \dd{\sem{s}_{\Gamma, x : \tau}}
  \circ \pair{\Id{}}{\pair{0}{\sem{t}_\Gamma \circ \pi_1}}$. Or, using the
  notation in Definition~\ref{def:star-composition}, $\sem{\diffS{s}{x}{t}} =
  \sem{s} \star \sem{t}$.
  \end{enumerate}
\end{lem}
\begin{proof}
  The proof follows roughly the structure of
  \cite[Theorem~4.11]{bucciarelli2010categorical}, taking into account the
  differences in our notion of differential substitution. Note also that we
  prove substitution in the case that the variable $x$ is not free in $t$. This
  is because we require this (stronger) form of substitution to write
  $\subst{e}{x}{x + \varepsilon(t)}$ in some cases of differential substitution.
  
  Both properties will follow by induction on the typing derivation of $s$. The
  only non-trivial case for the first one is differential application. For this,
  we must show that $\sem{\D(\subst{s}{x}{t}) \cdot (\subst{u}{x}{t})}$ is equal
  to $\sem{\D(s) \cdot u} \circ \pair{\Id{}}{\sem{t}}$. Expanding the term we obtain:
  \begin{align*}
    \sem{\D(\subst{s}{x}{t}) \cdot (\subst{u}{x}{t})}
    &= \Lambda \Big(
      \Lambda^-(\sem{\subst{s}{x}{t}}) \star \sem{\subst{u}{x}{t}}
    \Big)
    \displaybreak[0]\\
    &= \Lambda \Big(\Lambda^-(\sem{s} \circ \pair{\pi_1}{\pair{0}{\sem{u}}})
     \star
      (\sem{u} \circ \pair{\pi_1}{\pair{0}{\sem{u}}})
    \Big)
  \end{align*}
  By Lemma~\ref{lem:lambda-star-ev}(iii.), the above expression can be written
  as:
  \[ \Lambda \Big( (\Lambda^-\sem{s}) \star \sem{u} \Big) \circ \pair{\pi_1}{\sem{t}} \]
  which concludes the proof.

  We show now the cases for differential substitution.
  \begin{itemize}
    \item $s = x$
      Then $\sem{s} = \pi_2$ and 
      \begin{align*}
        \hspace{-30pt}
        \dd{\sem{s}} \circ \pair{\Id{}}{\pair{0}{\sem{t} \circ \pi_1}}
        = \pi_{22}
        \circ \pair{\Id{}}{\pair{0}{\sem{t} \circ \pi_1}}
        = \sem{t} \circ \pi_1
        = \sem{\Gamma, x : \tau \ts t : \tau}
      \end{align*}

    \item $s = y \neq x$
    
      Then $\sem{s} = \pi_2 \circ \pi_1^n \circ \pi_1$ and
      \begin{align*}
        \hspace{-30pt}
        \dd{\sem{s}} \circ \pair{\Id{}}{\pair{0}{\sem{t} \circ \pi_1}}
        = \pi_{2} \circ \pi_1^n \circ \pi_1 \circ \pi_2
        \circ \pair{\Id{}}{\pair{0}{\sem{t} \circ \pi_1}}
        = 0
        = \sem{\Gamma, x : \tau \ts 0 : \tau}
      \end{align*}
    
    \item $s = \epsilon s_1$

      Then $\sem{s} = \epsilon\sem{s_1}$ and
      \begin{align*}
        \hspace{-30pt}
        \dd{\sem{s}} \circ \pair{\Id{}}{\pair{0}{\sem{t} \circ \pi_1}}
        = \epsilon (\dd{\sem{s_1}} \circ \pair{\Id{}}{\pair{0}{\sem{t} \circ \pi_1}})
        = \epsilon \sem{\diffS{s_1}{x}{t}}
        = \sem{\diffS{(\epsilon s_1)}{x}{t }}
      \end{align*}

    \item The case $s = \sum_{i=1}^n s_i$ follows by a similar argument as the
    previous one.
    
    \item $s = \lambda y . s_1 : \sigma_1 \Ra \sigma_2$

      Then $\Gamma, x : \tau, y : \sigma_1 \ts s_1 : \sigma_2$ and therefore,
      by the induction hypothesis, we know that:
      \begin{align*}
        \hspace{-30pt}
        \sem{\diffS{s_1}{x}{t}}_{\Gamma, x : \tau, y : \sigma_1}
        &=
        \sem{\diffS{s_1}{x}{t}}_{\Gamma, y : \sigma_1, x : \tau}
        \circ \A{sw}\\
        &=
        \dd{\sem{s_1}_{\Gamma, y : \sigma_1, x : \tau}}
        \circ \pair{\Id{}}
                   {\pair{0}{\sem{t}_{\Gamma, y:\sigma_1, x:\tau}}}
        \circ \A{sw}\\
        &=
        \dd{\sem{s_1}_{\Gamma, x : \sigma_1, y : \tau}}
        \circ (\A{sw} \times \A{sw})
        \circ \pair{\Id{}}
                   {\pair{0}{\sem{t}_{\Gamma} \circ \pi_{11}}}
        \circ \A{sw}
        \displaybreak[0]\\
        &=
        \dd{\sem{s_1}_{\Gamma, x : \sigma_1, y : \tau}}
        \circ \pair{\A{sw}}
                   {\pair{\pair{0}{\sem{t}_\Gamma \circ \pi_{11}}}{0}}
        \circ \A{sw}
        \displaybreak[0]\\
        &=
        \dd{\sem{s_1}_{\Gamma, x : \sigma_1, y : \tau}}
        \circ \pair{\Id{}}
                   {\pair{\pair{0}{\sem{t}_\Gamma \circ \pi_{11}}}{0}}
      \end{align*}
      Obtaining the final result is just a matter of applying this identity
      and writing $\dd{\Lambda(\sem{s_1)}}$ in terms of the swapping map $\A{sw}$
      as remarked in Definition~\ref{def:difference-lambda-category}.
      \begin{align*}
        \dd{\sem{s}} \mathrel{\circ}{}& \pair{\Id{}}{\pair{0}{\sem{t} \circ \pi_1}}
        =\dd{\Lambda(\sem{s_1})} \circ \pair{\Id{}}{\pair{0}{\sem{t} \circ \pi_1}}
        \displaybreak[0]\\
        &= \Lambda \Big[
          \dd{\sem{s_1}}
          \circ \pa{\Id{} \times \pair{\Id{}}{0}}
          \circ \A{sw}
        \Big]
        \circ \pair{\Id{}}{\pair{0}{\sem{t} \circ \pi_1}}
        \displaybreak[0]\\
        &= \Lambda \Big[
          \dd{\sem{s_1}}
          \circ \pa{\Id{} \times \pair{\Id{}}{0}}
          \circ \A{sw}
          \circ \pair{\pair{\Id{}}{\pair{0}{\sem{t} \circ \pi_1}} \circ \pi_1}{\pi_2}
        \Big]
        \displaybreak[0]\\
        &= \Lambda \Big[
          \dd{\sem{s_1}}
          \circ \pa{\Id{} \times \pair{\Id{}}{0}}
          \circ \A{sw}
          \circ \pair{\pair{\pi_1}{\pair{0}{\sem{t} \circ \pi_{11}}}}{\pi_2}
        \Big]
        \displaybreak[0]\\
        &= \Lambda \Big[
          \dd{\sem{s_1}} 
          \circ \pa{\Id{} \times \pair{\Id{}}{0}}
          \circ \pair{\Id{}}{\pair{0}{\sem{t} \circ \pi_{11}}}
        \Big]
        \displaybreak[0]\\
        &= \Lambda \Big[
          \dd{\sem{s_1}} 
          \circ \pair{\Id{}}{\pair{\pair{0}{\sem{t} \circ \pi_{11}}}{0}}
        \Big]
        \displaybreak[0]\\
        &= \Lambda \Big[
          \dd{\sem{s_1}} 
          \circ (\A{sw} \times \A{sw})
          \circ \pair{\Id{}}{\pair{0}{\sem{t} \circ \pi_{11}}}
        \Big]
        \displaybreak[0]\\
        &= \Lambda \Big[
          \dd{\sem{s_1}} 
          \circ \pair{\A{sw}}{\pair{\pair{0}{\sem{t} \circ \pi_{11}}}{0}}
        \Big]
        \displaybreak[0]\\
        &= \Lambda \Big[
          \dd{\sem{s_1}} 
          \circ \pair{\Id{}}{\pair{\pair{0}{\sem{t} \circ \pi_{11}}}{0}}
          \circ \A{sw}
        \Big]
        \displaybreak[0]\\
        &= \Lambda \sem{\Gamma, x : \tau, y : \sigma_1 \ts 
          \diffS{s_1}{x}{t} : \sigma_2
        }
        \displaybreak[0]\\
        &= \sem{\Gamma, x : \tau \ts 
          \lambda y . \diffS{s_1}{x}{t} : \sigma_1 \Ra \sigma_2
        }
      \end{align*}

      \newcommand{\theterm}[0]{\mathfrak{s}}
    \item $s = s_1\ e$
      To simplify the calculations, we write $\theterm$ for
      $\Lambda^-(\sem{s_1})$. 
      The result follows as
      a consequence of Lemma~\ref{lem:lambda-star-ev}[iii.].
      \begin{align*}
        \lefteqn{\sem{\diffS{(s_1\ e)}{x}{t}}}\\
        &= \sem{
            \D(s_1) \cdot \pa{\diffS{e}{x}{t}}\ e
        } + \sem{
          \pa{\diffS{s_1}{x}{t}}\ (\subst{e}{x}{x + \varepsilon t})
        }
        \displaybreak[0]\\
        &= \A{ev} \circ \pair{\Lambda\pa{
          \theterm \star (\sem{e} \star \sem{t})
        }}{\sem{e}}\\
        \continued + 
          \A{ev} \circ \pair{
            \sem{s_1} \star \sem{t}
          }{
            \sem{e} \circ \pair{\pi_1}{\pi_2 + \varepsilon(\sem{t})\circ \pi_1}
          }
        \displaybreak[0]\\
        &= (\A{ev} \circ \pair{\sem{s_1}}{\sem{e}}) \star \sem{t}
        \displaybreak[0]\\
        &= \sem{s_1\ e} \star \sem{t}
      \end{align*}

    \item $s = \D(s_1)\cdot u$

      We will again abbreviate $\Lambda^-(\sem{s_1})$ as $\theterm$ to make the
      subsequent calculations more readable.
      The result then follows
      by applying of Lemma~\ref{lem:lambda-star-ev}[ii.].
      \begin{align*}
        \lefteqn{\sem{\diffS{\D(s_1) \cdot u}{x}{t}}}\\
        &= \sem{
          \D(s_1) \cdot \diffS{u}{x}{t}
        }
        + \sem{
          \D\pa{\diffS{s_1}{x}{t}} \cdot (\subst{u}{x}{x + \epsilon(t)})
        }\\
        \continued + \epsilon\sem{
          \pa{\D(\D(s_1) \cdot u) \cdot \pa{\diffS{u}{x}{t}}}
        }
        \displaybreak[0]\\
        &= \Lambda\pa{
          \theterm \star \sem{\diffS{u}{x}{t}}
        }
        + \Lambda\pa{
          \pa{\Lambda^-\sem{\diffS{s_1}{x}{t}}}
          \star {\sem{\subst{u}{x}{x + \varepsilon(t)}}}
        }\\
        \continued + \varepsilon\Lambda\pa{
          \pa{\Lambda^-\sem{\D(s_1) \cdot u}} \star \sem{\diffS{u}{x}{t}}
        }
        \displaybreak[0]\\
        &= \Lambda\pa{
          \theterm \star (\sem{u} \star \sem{t})
        }
        + \Lambda\pa{
          \pa{\Lambda^-(\Lambda(\sem{s_1}) \star \sem{t})}
          \star {\sem{\subst{u}{x}{x + \varepsilon(t)}}}
        }\\
        \continued + \varepsilon\Lambda\pa{
          \pa{\Lambda^-(\theterm \star \sem{u})} \star (\sem{u} \star \sem{t})
        }
        \displaybreak[0]\\
        &= \Lambda\pa{
          \theterm \star (\sem{u} \star \sem{t})
        }
        + \Lambda\pa{
          \pa{\Lambda^-(\Lambda(\sem{s_1}) \star \sem{t})}
          \star \pa{{\sem{u} \circ \pair{\pi_1}{\pi_2 + \varepsilon(\sem{t}) \circ \pi_1}}}
        }\\
        \continued + \varepsilon\Lambda\pa{
          \pa{\Lambda^-(\theterm \star \sem{u})} \star (\sem{u} \star \sem{t})
        }
        \displaybreak[0]\\
        &= \pa{\Lambda\pa{\theterm \star \sem{u}}} \star \sem{t}
        \displaybreak[0]\\
        &= \pa{\Lambda\pa{\Lambda^-(\sem{s_1}) \star \sem{u}}} \star \sem{t}
        \displaybreak[0]\\
        &= (\sem{\D(s_1) \cdot u}) \star \sem{t}
      \end{align*}
  \end{itemize}
\end{proof}

\begin{defi}
  Given well-formed terms $\ul{s}, \ul{s'}$, we define the equivalence relation
  $\sim_{\beta\partial}$ as the least contextual equivalence relation that contains
  the one-step reduction relation $\wfReducesTo$.
\end{defi}

\begin{cor}
  The interpretation $\sem{\cdot}$ is \emph{sound}, that is to say, whenever
  $\ul{s} \sim_{\beta\partial} \ul{s'}$ then $\sem{s} = \sem{s'}$, independently of
  the choice of representatives $s, s'$.
\end{cor}

\newcommand\diflfunc{\mathsf{Dif{\lambda}\hbox{-}Func}}
\newcommand\modcat{\mathsf{Mod}_{\mathsf{Dif}\lambda}}

\begin{defi}
  Recall that a simply-typed theory is a collection of equational judgements
  of the form $\Gamma \ts s = t : \sigma$ where $\Gamma \ts s : \sigma$ and
  $\Gamma \ts t : \sigma$ are derivable. We say that a simply-typed theory is a
  \emph{difference $\lambda$-theory} if it is closed under all rules in the system
  $\lambda_\epsilon^{\times} \beta \eta \partial$ (comprising the contextual rules for all the constructs of the $\lambda_\epsilon$-calculus augmented by products, $\sim_\epsilon$ equivalence, and the surjective pairing, $\beta$, $\eta$ and $\partial$ laws, and last being the equational version of $\reducesTo_\partial$). 

  Given an interpretation $\sem{\cdot}_{\mathcal M}$ of $\lambda_\epsilon$ in
  $\cat{C}$, we say that ${\mathcal M} = \sem{\cdot}_{\mathcal M}$ is a \emph{model} of
  a difference $\lambda$-theory $\mathscr{T}$ if for every typed equational
  judgement $\Gamma \ts s = t : \sigma$ in $\mathscr{T}$, we have that
  $\sem{\Gamma \ts s : \sigma}_{\mathcal M}$ and $\sem{\Gamma \ts t : \sigma}_{\mathcal
  M}$ are the same morphism. 
  
  A \emph{model homomorphism} $h : {\mathcal M}
  \to {\mathcal N}$ is given by isomorphisms $h_{\mathbf{t}} : \sem{\mathbf{t}}_{\mathcal
  M} \to \sem{\mathbf{t}}_{\mathcal N}$ for each basic type $\mathbf{t}$, and
  $h_{\sigma \times \tau} \defeq h_{\sigma} \times h_{\tau}$, and
  \[
  h_{\sigma \Ra \tau} \defeq h^{-1}_\sigma \Ra h_\tau \defeq \Lambda(h_\tau \circ \A{ev} \circ (\Id{} \times h^{-1}_\sigma)). 
  \]

  We write $\modcat({\mathscr{T}}, \cat{C})$ for the category whose objects are
  all models of difference $\lambda$-theory $\mathscr{T}$ in a difference
  $\lambda$-category $\cat{C}$, and whose morphisms are model homomorphisms. 
\end{defi}

\begin{defi}
  Let $\cat C$ and $\cat D$ be difference $\lambda$-categories. We say that a
  functor $F : \cat C \to \cat D$ is a \emph{difference $\lambda$-functor} if
  $F$ preserves the following:
  \begin{itemize}
    \item additive structure: $F(f + g) = F(f) + F(g)$, and $F(0) = 0$
    \item infinitesimal extension: $F (\epsilon (f)) = \epsilon (F (f))$
    \item products via the isomorphism $\Phi \defeq \pair{F (\pi_1)}{F (\pi_2)}$
    \item exponentials via the isomorphism $\Psi \defeq \Lambda (F (\A{ev}) \circ \Phi)$
    \item difference combinator: $F (\dd{f}) = \dd{F(f)} \circ \Phi$.
  \end{itemize}
  We write $\diflfunc(\cat{C}, \cat{D})$ for the category of {difference
  $\lambda$-functors} $\cat{C} \to \cat{D}$ and natural isomorphisms. 
\end{defi}
 
\begin{defi}
  Given a difference $\lambda$-theory $\mathscr{T}$, we say that a category,
  denoted $\mathbf{Cl}({\mathscr{T}})$, is \emph{classifying} if there is a model
  of the theory in $\mathbf{Cl}({\mathscr{T}})$, and this model is ``generic'',
  meaning that for every differential $\lambda$-category $\cat{D}$, there is a
  natural equivalence
  \begin{equation}
  \diflfunc(\mathbf{Cl}({\mathscr{T}}), \cat{D}) 
  \simeq
  \modcat({\mathscr{T}}, \cat{D}).
  \label{eq:classifying}
  \end{equation}
\end{defi}

The classifying category (unique up to isomomrphism) is the ``smallest'' in the sense that given a model of the theory $\sem{\cdot}_{\cat{D}}$ in a difference $\lambda$-category $\cat{D}$,
there is a difference $\lambda$-functor $F : \mathbf{Cl}(\mathscr{T}) \to \cat{D}$ such that the interpretation $\sem{\cdot}_{\cat D}$ can be factored through the canonical interpretation in the classifying category, i.e., $\sem{\cdot}_{\cat D} = F \circ \sem{\cdot}_{\mathbf{Cl}(\mathscr{T})}$.

\begin{conjecture}[Completeness] 
  Every difference $\lambda$-theory $\mathscr{T}$
  has a classifying difference $\lambda$-category $\mathbf{Cl}({\mathscr{T}})$.
\end{conjecture}

\section{Conclusions and Future Work}

We have defined here the difference $\lambda$-calculus, which generalises the
differential $\lambda$-calculus in exactly the same manner as Cartesian
difference categories generalise their differential counterpart. While this
calculus is of theoretical interest, it lacks most practical features, such as
iteration or conditionals, and it is not immediately obvious how to extend it
with these. It is not clear, for example, precisely when iteration combinators
are differentiable in the difference category sense.

The problem of iteration is closely related to integration, which is itself the
focus of current work on the differential side
\cite{cockett2019integral,lemay2019exponential}. Indeed, consider a hypothetical
extension of the difference $\lambda$-calculus equipped with a type of natural
numbers (with the identity as its corresponding infinitesimal extension, that is
to say, $\epsilon_\NN = \Id{\NN}$). How should an iteration operator
$\textbf{iter}$ be defined? The straightforward option would be to give it the
usual behavior, that is to say:
\[
\begin{array}{rcl}
  \textbf{iter}\ \A{Z}\ z\ s &\reducesTo & z\\
  \textbf{iter}\ (\A{S}\ n)\ z\ s &\reducesTo & s\ (\textbf{iter}\ n\ z\ s)
\end{array}
\]
These reduction rules entail that every object involved must be complete, that
is to say, for every $s, t : A$, there is some $u : A$ with $s \mathrel{+}{}
\epsilon(u) = t$ -- such an element is given by the term $((\D (\lambda n .
\textbf{iter}\ n\ s\ (\lambda x . t)) \cdot (\A{S}\ \A{Z}))\ \A{Z})$.

This would rule out a number of interesting models and so it seems
unsatisfactory. An alternative is to define the iteration operator by:
\[
\begin{array}{rcl}
  \textbf{iter}\ \A{Z}\ z\ s &\reducesTo & z\\
  \textbf{iter}\ (\A{S}\ n)\ z\ s &\reducesTo & (\textbf{iter}\ n\ z\ s) 
  + \epsilon(s\ (\textbf{iter}\ n\ z\ s))
\end{array}
\]
Fixed $z, s$, and defining the map $\mu(n) \defeq \textbf{iter}\ n\ z\ s$, its
derivative $\D[\mu](n, \A{S}\ \A{Z})$ is precisely $s(\mu(n))$. Or, in other
words, the function $\mu : \NN \to A$ is a ``curve'' which starts at $z$
and whose derivative at a given point $n$ is $s(\mu(n))$ -- this boils down to
stating that the curve $\mu$ is an integral curve for the vector field $s$
satisfying the initial condition $\mu(\A{Z}) = z$! Hence it may be possible to
understand iteration as a discrete counterpart of the Picard-Lindel\"of theorem,
which states that such integral curves always exist (locally).

It would be of great interest to extend $\lambda_\epsilon$ with an interation
operator and give its semantics in terms of differential (or difference)
equations. Studying recurrence equations using the language of differential
equations is a very useful tool in discrete analysis; for example, one can treat
the recursive definition of the Fibonacci sequence as a discrete ODE and use
differential equation methods to find a closed-form solution. We believe that in
a language which frames iteration in such terms may be amenable to optimisation
by similar analytic methods.

 \bibliographystyle{spmpsci}
 \bibliography{biblio}
\end{document}